\documentclass[]{article}
\usepackage{url}
\usepackage{setspace} 

\usepackage[a4paper]{geometry}
\usepackage[utf8]{inputenc}
\usepackage{lipsum}
\usepackage{natbib}
\usepackage{pbox}
\usepackage{enumitem}
\usepackage{graphicx}
\usepackage{esvect}
\usepackage{amsthm,amsmath,bm,dsfont,amssymb,multirow}
\usepackage[export]{adjustbox}
\usepackage[figuresright]{rotating}
\usepackage{natbib}
\usepackage{pdflscape}
\usepackage{multirow}

\usepackage{tikz}
\usetikzlibrary{positioning}

\newcommand{\norm}[1]{\left\lVert#1\right\rVert}
\newcommand{\GGM}{GGM}
\newcommand{\RCON}{RCON}

\newcommand{\PDRCON}{pdRCON}
\newcommand{\pdglasso}{pdglasso}
\newcommand{\G}{\mathcal{G}}

\newcommand{\lambdadiag}{\lambda^{\texttt{diag}}_{1}}
\newcommand{\lambdasym}{\lambda^{\texttt{sym}}_{2}}

\DeclareMathOperator{\tr}{tr}
\DeclareMathOperator{\diag}{diag}
\DeclareMathOperator{\vd}{vd}
\DeclareMathOperator{\mle}{mle}
\DeclareMathOperator*{\argmin}{arg\,min}

\DeclareMathOperator{\vech}{vech}

\DeclareMathOperator{\myvec}{v}

\usepackage{color}

\usepackage{amsthm}
\theoremstyle{plain}
\newtheorem{thm}{Theorem}[section]

\newtheorem{prop}[thm]{Proposition}

\theoremstyle{definition}

\newtheorem{example}{Example}[section]
\newenvironment{exmp}{\begin{example}}{\hfill\qedsymbol\end{example}}

\theoremstyle{remark}

\begin{document}

\title{On the application of Gaussian graphical models to paired data problems}

\author{Saverio Ranciati$^\ast$\\[4pt]
\textit{Department of Statistical Sciences,
University of Bologna,
Italy}\\[2pt]
{saverio.ranciati2@unibo.it}\\[6pt]
Alberto Roverato\\[4pt]
\textit{
Department of Statistical Sciences,
University of Padova,
Italy}\\[2pt]
{alberto.roverato@unipd.it}
}
\maketitle

\footnotetext{To whom correspondence should be addressed.}

\begin{abstract}
Gaussian graphical models are nowadays commonly applied to the comparison of groups sharing the same variables, by jointly learning their independence structures. We consider the case where there are exactly two dependent groups and the association structure is
represented by a family of coloured Gaussian graphical models suited to deal with paired data problems. To learn the two dependent graphs, together with their across-graph association structure, we implement a fused graphical lasso penalty. We carry out a comprehensive analysis of this approach, with special attention to the role played by some relevant submodel classes. In this way, we provide a broad set of tools for
the application of Gaussian graphical models to paired data problems. These include results useful for the specification of penalty values in order to obtain a path of lasso solutions and an ADMM algorithm that solves the fused graphical lasso optimization problem. Finally, we carry out a simulation study to compare our method with the traditional graphical lasso, and present an application of our method to cancer genomics where it is of interest to compare cancer cells with a control sample from histologically normal tissues adjacent to the tumor. All the methods described in this article are implemented in the \textsf{R} package \texttt{pdglasso} available at \url{https://github.com/savranciati/pdglasso}.
\end{abstract}

\section{Introduction}\label{SEC:introduction}
Graphical models are powerful tools for expressing the relationships between variables. In Gaussian graphical models (\GGM{s}) the dependency structure is obtained by associating an undirected graph to the concentration matrix, that is the inverse of the covariance matrix. The graph has one vertex for every variable and every missing edge implies that the corresponding entry of the concentration matrix is equal to zero; see \cite{lauritzen1996graphical}.

In recent years, there has been a great deal of interest in the joint learning of multiple networks, where the observations come from two or more groups sharing the same variables. For instance, \citet{danaher2014joint} considered the case of gene expression measurements collected from the cancer tissue of a sample of patients and from the normal tissue of a control sample. Hence, the association structure of each group is represented by a network and methods for the joint learning have been developed to deal with the fact that
networks are expected to share similar patterns while retaining individual features. In this framework, the literature has mostly focused on the case where the groups are independent so that every network is a distinct unit, disconnected from the other networks; see \citet{tsai2022joint} for a review. Thus, specific methods that deal with the across network association are required in the case where the groups are dependent. \citet{xie2016joint} considered the case of gene expression data obtained from multiple tissues from the same individual and modelled the cross-graph dependence between groups by means of a latent vector representing systemic variation  manifesting simultaneously in all groups; see also \citet{zhang2022bayesian}. \citet{roverato2022modelinclusion,roverato2024exploration} focused on paired data problems where there are exactly two dependent groups. Paired data are commonly originated from experimental designs where each subject is measured twice under two different conditions, or time points, as well as in matched pairs designs. For instance, in cancer genomics it is common practice to take control samples from histologically normal tissues adjacent to the tumor \citep{aran2017comprehensive}. \citet{roverato2022modelinclusion} approached the problem by considering a family of coloured \GGM{s} \citep{hojsgaard2008graphical}, named \PDRCON\ models, in which similarities between groups are represented by symmetries which can involve both the graph structure and the concentration values, in the form of equality constraints. One of the appealing features of this approach is that the resulting model has a graph for each of the two groups with the cross-graph dependence explicitly represented by the edges joining one group with the other, which may themselves present symmetries; we refer to Section~\ref{SEC:coloured.ggms} for details. In a related framework,
\citet{ranciati2021fused} considered the task of learning brain networks from fMRI data, and approached the problem by means of a fused graphical lasso procedure, named the symmetric graphical lasso, designed with the specific aim to identify symmetries between the left and the right hemisphere.

In this paper, we introduce the \emph{graphical lasso for paired data} (\pdglasso) that extends the symmetric graphical lasso to deal with the wider family of \PDRCON\ models, and then provide a set of methods which are meant to establish an extensive collection of tools for the immediate application of \GGM{s} to the analysis of paired data. We provide an alternating directions method of multiplier (ADMM) algorithm that solves the \pdglasso\ optimization problem and, furthermore, the \textsf{R}  package \citep{Rmanual2023} \texttt{pdglasso} that implements the ADMM algorithm as well as a function for the computation of maximum likelihood estimates, and other utility functions to deal efficiently with the objects resulting from the analysis. For details about the package we refer to Section `Code and data availability' of this paper. We analyse the role played by relevant submodel classes characterized by a fully symmetric structure. More specifically, full symmetry can be a property of the overall structure, or it may be confined to some specific components, such as the diagonal entries of the concentration matrix, the inside-group structure or the across-group structure. We provide results on the values of the penalty terms required to obtain either a diagonal, a block-diagonal or a fully symmetric solution, which are useful in the specification of the grid of penalty values required to obtain a path of \pdglasso\ solutions.  We then carry out a simulation study that shows that our pdglasso method has comparable performances to glasso when there are no symmetries, and improves on the latter when symmetries are present. Furthermore, we present an application to the analysis of gene expression data in cancer genomics.

The rest of this paper is organized as follows. Section~\ref{SEC:coloured.ggms} provides the background on coloured \GGM{s} and on graphical models for paired data, as required for this paper. The \pdglasso\ problem is introduced in
Section~\ref{SEC:pdglasso}, whereas Section~\ref{SEC:pdglasso.solutions} deals with the practical application of the method, including the maximum theoretical values of the penalty parameters. Section~\ref{SEC:submodel.classes} is devoted to some relevant families of submodels. The ADMM algorithm for the optimization of the penalized likelihood can be found in Section~\ref{SEC:appendix.algo}. The simulation study and the application to gene expression data are described in Section~\ref{SEC:simulations} and \ref{SEC:application}, respectively. Finally, Section~\ref{SEC:conclusions} contains a brief discussion. Proofs are deferred to Appendix~\ref{SEC:proofs}, whereas the remaining three appendices contain a comprehensive presentation of fully symmetric models, some illustrative examples, and some additional material concerning the application to gene expression data, respectively.

\section{Coloured Gaussian graphical models for paired data}\label{SEC:coloured.ggms}

Let $Y_{V}$ be a continuous random vector indexed by a finite set $V=\{1,\ldots, p\}$. We denote by $\Sigma=\{\sigma_{ij}\}_{i,j\in V}$ and $\Sigma^{-1}=\Theta=\{\theta_{ij}\}_{i,j\in V}$ the covariance and the concentration matrix of $Y_{V}$, respectively. Both $\Sigma$ and $\Theta$ belong to the set $\mathcal{S}^{+}_{p}$ of (symmetric) $p\times p$ positive definite matrices and we assume, without loss of generality, that $Y_{V}$ has zero mean vector.  An undirected graph with vertex set $V$ is a pair $\G=(V, E)$ where $E$ is an edge set that is a set of pairs of distinct vertices and, with a slight abuse of notation, we will sometimes write $\{i,j\}\in \G$, in place of $\{i,j\}\in E$, to mean that the edge $\{i,j\}$ belongs to $E$. We say that the concentration matrix $\Theta$ is adapted to a graph $\G=(V, E)$ if every missing edge of $\G$ corresponds to a zero entry in $\Theta$; formally, $\{i,j\}\not\in E$ implies $\theta_{ij}=0$ for every $i,j\in V$ with $i\neq j$. A Gaussian graphical model (\GGM) with graph $\G$ is the family of multivariate normal distributions for $Y_{V}$ whose concentration matrix is adapted to $\G$. These models are also known with the name of covariance selection models or concentration graph models \citep{lauritzen1996graphical}.

In paired data problems, the set $V$ is naturally partitioned into a $L$eft and a $R$ight block, $V=L\cup R$ with $|L|=|R|=q=p/2$ and every variable in $Y_{L}$ has a corresponding variable in $Y_{R}$. We set $i^{\prime}=i+q$ for every $i\in L$ and, without loss of generality, we index the variables so that $Y_{i}$ corresponds to $Y_{i^{\prime}}$ for every $i\in L$. In this way, it also holds that $L=\{1,\ldots, q\}$ and $R=\{q+1,\ldots, p\}$. Accordingly, the concentration matrix $\Theta$ can be naturally partitioned into four blocks,
\begin{align}\label{EQN:Theta.blocks}
    \Theta =
    \left(
    \begin{array}{cc}
    \Theta_{LL} & \Theta_{LR}\\
    \Theta_{RL} & \Theta_{RR}\\
    \end{array}
    \right).
\end{align}

The  subgraph of $\G$ induced by a subset $A\subseteq V$ is denoted by $\G_{A}=(A, E_{A})$, where $\{i,j\}\in E_{A}$ if and only if both $i,j\in A$ and $\{i,j\}\in E$. If $\Theta$ is adapted to $\G$ then it holds that $\Theta_{LL}$ and $\Theta_{RR}$ are adapted to $\G_{L}$ and $\G_{R}$, respectively. Thus, $\G_{L}$ and $\G_{R}$ are the group specific graphs and the interest is for similarities involving the independence structures of the two groups, that we call \emph{structural symmetries}. We distinguish between two types of structural symmetries. When for a pair $i,j\in L$, with $i\neq j$  the edges $\{i,j\}$ and $\{i^{\prime},j^{\prime}\}$ are either both present or both missing in $\G_{L}$ and $\G_{R}$, respectively, we say that there is an \emph{inside-block} structural symmetry. On the other hand, $Y_{L}$ and $Y_{R}$ are not expected to be independent, and symmetries can also appear in the cross-group association structure, thereby involving the edges connecting the vertices in $L$ with the vertices in $R$. Hence, we say that an \emph{across-block} structural symmetry is present if the edges $\{i, j^{\prime}\}$ and $\{i^{\prime}, j\}$ are either both present or both missing in $\G$.

In a \GGM\ every  entry of the concentration matrix can be naturally associated with either a vertex or an edge of the graph and, thus,
structural symmetries in the graph imply that the corresponding entries in the concentration matrix are both either zero or non-zero. Hence, a stronger class of models could be defined  that specifies equality constraints between such entries. Equality constraints encode stronger similarities and result in more parsimonious models. From this viewpoint, every structural symmetry due to a pair of missing edges also corresponds to a \emph{parametric symmetry} because if, for example, both
$\{i,j\}\not\in \G_{L}$ and $\{i^{\prime},j^{\prime}\}\not\in \G_{R}$, then also the associated parameters have the same value, $\theta_{ij}=\theta_{i^{\prime}j^{\prime}}=0$. This idea can be extended to symmetric pairs of non-missing edges, and we say that there is an inside-block parametric symmetry if for a pair $i,j\in L$, with $i\neq j$, it holds that $\theta_{ij}=\theta_{i^{\prime}j^{\prime}}$. On the other hand, across-block parametric symmetries are also possible in the case where, for some $i,j\in L$, it holds that $\theta_{ij^{\prime}}=\theta_{i^{\prime}j}$, with $i\neq j$. Finally, parametric symmetries can also be satisfied by the diagonal entries of $\Theta$ and we say that there is a \emph{vertex parametric symmetry} if for $i\in L$ it holds that $\theta_{ii}=\theta_{i^{\prime}i^{\prime}}$.
Appendix~\ref{SEC:appendix.examples} gives some examples of \PDRCON\ models with a detailed description of the different types of symmetries they may represent.

Gaussian graphical models with additional restrictions on the parameter space were introduced in the seminal paper  by \citet{hojsgaard2008graphical} with the name of coloured \GGM{s}. The family of coloured \GGM{s} characterized by equality \emph{R}estrictions on \emph{CON}centration values are known as \RCON\ models, whereas the subfamily of \RCON\ models comprising the equality constraints suited for paired data problems, as described above, was introduced by \citet{roverato2022modelinclusion} with the name of \RCON\ models for paired data (\PDRCON\ models); see also \citet{ranciati2021fused} for an application of \RCON\ models for paired data that only involves inside-block and vertex symmetries.

Coloured \GGM{s} are typically represented by coloured graphs where vertices and edges depicted in black identify unconstrained parameters, whereas other colours are used to identify subsets of parameters which are constrained to having the same value. In this way, a different colour is required for every distinct equality constraint. We note, however, that this representation is redundant in the case of \PDRCON\ models, because equality constraints may only involve specific pairs of parameters. Indeed, in order to identify the model from the graph it is sufficient to be able to distinguish whether a parameter is constrained or not. It follows that a \PDRCON\ model can be unambiguously identified also if all the non-black vertices and edges are depicted by the same colour. This makes the graphical representation more readable for large graphs and, thus, here we will follow this rule and simply refer to coloured vertices and edges in contrast to uncoloured (i.e. black) ones. In Appendix~\ref{SEC:appendix.examples}, in order to make graphs readable also in black and white printing, we distinguish between uncoloured and coloured edges by using thin and thick lines, respectively.
\section{The graphical lasso for paired data}\label{SEC:pdglasso}
For a sample $y^{(1)}_{V},\ldots,y^{(n)}_{V}$ of i.i.d. observations of $Y_{V}\sim N(0, \Sigma)$, the maximum likelihood estimator of $\Theta$ is the value that maximises  the log-likelihood function
\begin{align}\label{EQN:log-lik}
l(\Theta)=\log\det(\Theta)-\tr(S\Theta)
\end{align}
over $\mathcal{S}^{+}_{p}$ where $S=\{s_{ij}\}_{i,j\in V}$ is the matrix $S=n^{-1}\sum_{i=1}^{n}y^{(i)}_{V}(y^{(i)}_{V})^{\top}$. In \GGM{s} the interest is for the zero pattern of $\Theta$ and \citet{yuan2007model} proposed an estimator of $\Theta$ obtained from the minimization of the penalized log-likelihood
\begin{align}\label{EQN:glasso}
\mathcal{L}_{\lambda_{1}}(\Theta)=-l(\Theta)+\mathcal{P}_{\lambda_{1}}(\Theta)
\end{align}
with $\mathcal{P}_{\lambda_{1}}(\Theta)=\lambda_{1}\norm{\Theta}_{1}$, where
$\norm{\cdot}_{1}$ denotes the $\ell_{1}$-norm, that is the sum of the absolute values of the entries of the matrix, and $\lambda_{1}$ is a nonnegative regularization parameter. The minimization of (\ref{EQN:glasso}) over $\mathcal{S}^{+}_{p}$ is known as the \emph{graphical lasso} problem and a number of algorithms have been proposed for its solution; see, among others, \citet{friedman2008sparse} and \citet[][Section~6.5]{boyd2011distributed}. Unlike the maximum likelihood estimator, for $\lambda_{1}>0$ the solution to the graphical lasso exists also when $p>n$, in which case $S$ is singular. The high popularity of the graphical lasso is due to the fact that it simultaneously performs estimation and model selection within the family of \GGM{s}. Indeed, due to the geometry of the $\ell_{1}$-penalty, some of the off-diagonal entries of the concentration matrix are shrunk to exactly zero, with a non-decreasing level of sparsity as $\lambda_{1}$ increases.

In order to simultaneously perform estimation and model selection within the family of \PDRCON\ models we introduce an additional penalty term
\begin{align}\label{EQN:fused-penalty}
\mathcal{Q}_{\lambda_{2}}(\Theta)=
\lambda_{2}\left(
\norm{\Theta_{LL}-\Theta_{RR}}_{1}+\norm{\Theta_{LR}-\Theta_{RL}}_{1}
\right)
\end{align}
to (\ref{EQN:glasso}) so as to obtain,
\begin{align}\label{EQN:pdglasso}
\mathcal{L}_{\lambda_{1},\lambda_{2}}(\Theta)=-l(\Theta)+\mathcal{P}_{\lambda_{1}}(\Theta)+\mathcal{Q}_{\lambda_{2}}(\Theta),
\end{align}
and estimate $\Theta$ by the \emph{paired data graphical lasso} (\pdglasso) estimator $\widehat{\Theta} =\argmin_{\Theta} \mathcal{L}_{\lambda_{1},\lambda_{2}}(\Theta)$, where minimization is taken over $\mathcal{S}^{+}_{p}$.

Note that $\mathcal{Q}_{\lambda_{2}}(\Theta)$ in (\ref{EQN:fused-penalty}) is a fused type lasso penalty \citep{tibshirani2005sparsity,hoefling2010path} with regularization parameter $\lambda_{2}\geq 0$. For $\lambda_{2}=0$ the \pdglasso\ estimator coincides with the graphical lasso estimator whereas positive values of $\lambda_{2}$ encourage parametric symmetries. More specifically,  the term
$\lambda_{2}\norm{\Theta_{LL}-\Theta_{RR}}_{1}=\lambda_{2}\sum_{i\in L}|\theta_{ii}-\theta_{i^{\prime}i^{\prime}}|+\lambda_{2}\sum_{i,j\in L; i\neq j}|\theta_{ij}-\theta_{i^{\prime}j^{\prime}}|$ encourages both vertex parametric symmetries and inside-block parametric symmetries, whereas
$\lambda_{2}\norm{\Theta_{LR}-\Theta_{RL}}_{1}=\lambda_{2}\sum_{i,j\in L}|\theta_{ij^{\prime}}-\theta_{i^{\prime}j}|$ encourages across-block parametric symmetries; we remark that, in the latter sum, $|\theta_{ij^{\prime}}-\theta_{i^{\prime}j}|=0$ whenever $i=j$ because $\Theta$ is a symmetric matrix. Hence, like the graphical lasso, the \pdglasso\ is a sparse estimate of $\Theta$ when $\lambda_{1}$ is large, and has many parametric symmetries when $\lambda_{2}$ is large.

The idea of using a fused penalty to induce equality constraints between concentration parameters is not new, and it was first introduced by \citet{danaher2014joint} in the context of joint learning of multiple \GGM{s} for independent samples. Subsequently,
\citet{ranciati2021fused} considered the case of two possibly non independent groups and the \pdglasso\ problem in
(\ref{EQN:pdglasso}) extends the work of \citet{ranciati2021fused} also including cross-group symmetries.

\section{Selection of penalty parameters and application issues}\label{SEC:pdglasso.solutions}
This section deals with the practical application of the method. Specifically, we provide the maximum theoretical values of the penalty parameters, give details on the computation of the eBIC criterion we use for the selection of the model and, finally, discuss the role played by the unit of measurement of the variables.

The application of the \pdglasso\ requires the initial definition of a grid of penalty parameter values to obtain a path of \pdglasso\ solutions. To this aim, it is useful to identify the values of $\lambda_{1}$ and $\lambda_{2}$ that return a diagonal and a fully symmetric solution, respectively, to be set as maximum values of the grid. Formally, we say that a \PDRCON\ model is fully symmetric if and only if both $\Theta_{LL}=\Theta_{RR}$ and $\Theta_{LR}=\Theta_{RL}$. We remark that the equality $\Theta_{LL}=\Theta_{RR}$ implies (i) full vertex parametric symmetry in the sense that $\theta_{ii}=\theta_{i^{\prime}i^{\prime}}$ for every $i\in L$ and (ii) both structural and parametric inside-block symmetry because $\theta_{ij}=\theta_{i^{\prime}j^{\prime}}$ for every $i,j\in L$ with $i\neq j$ also implies that $\theta_{ij}=0$ if and only if $\theta_{i^{\prime}j^{\prime}}=0$. Similarly, the equality  $\Theta_{LR}=\Theta_{RL}$ implies both structural and parametric across-block symmetry. Hence, we have the following,
\begin{thm}\label{THM:hat.fully.sym}
A sufficient condition for the solution $\widehat{\Theta}$ to the \pdglasso\ to be fully symmetric  is that $\lambda_{2}\geq\lambdasym$, where
$\lambdasym=
\max\left\{
|s_{ij}-s_{i^{\prime}j^{\prime}}|/2,\;
|s_{i^{\prime}j}-s_{ij^{\prime}}|/2;\;\; i,j\in L
\right\}$.
\end{thm}
\begin{proof}
See Appendix~\ref{SEC:appendix.proof.main.theorem}.
\end{proof}

The \pdglasso\ problem coincides with the graphical lasso in the case where $\lambda_{2}$ is set to zero, and it is a well-known result that if $\lambda_{1}\geq \max_{i,j\in V; i\neq j}|s_{ij}|$ then the solution to the graphical lasso is a diagonal matrix so that the selected graph is fully disconnected; see among others \citet[][Section~2.1]{mazumder2012exact}.  The following proposition shows that this is still true in the \pdglasso\ problem for any $\lambda_{2}\geq 0$.
\begin{prop}\label{THM:lambda.diag.solution}
A sufficient condition for the solution  $\widehat{\Theta}$ to the \pdglasso\ to be a diagonal matrix is that $\lambda_{1}\geq \lambdadiag$, where
$\lambdadiag = \max\left\{|s_{ij}|;\;\; i,j\in V\mbox{ with }i\neq j\right\}$.

\end{prop}
\begin{proof}
See Appendix~\ref{SEC:proof.diagonal.sol}.
\end{proof}

For the sake of completeness, we also consider the case where the solution $\widehat{\Theta}$ to the \pdglasso\ is block diagonal with all the entries of $\widehat{\Theta}_{LR}$ equal to zero. This is of interest because $\widehat{\Theta}_{LR}=O_{q}$ implies that there are no across-block edges so that the two groups are independent. For the case where $\lambda_{2}=0$, i.e. in the graphical lasso, this problem was considered by \citet{witten2011new} and \citet{mazumder2012exact} where it is proved that such block diagonal structure is obtained whenever $\lambda_{1}\geq \max_{i,j\in L}|s_{ij^{\prime}}|$. Here, we show that the latter result is still true in the \pdglasso\ problem for any $\lambda_{2}\geq 0$.
\begin{prop}\label{THM:lambda.block.solution}
A sufficient condition for the solution $\widehat{\Theta}$ to the \pdglasso\ to be block diagonal with blocks $\widehat{\Theta}_{LL}$ and $\widehat{\Theta}_{RR}$ is that $\lambda_{1}\geq |s_{ij}|$ for every $i\in L$ and $j\in R$.
\end{prop}
\begin{proof}
See Appendix~\ref{SEC:proof.block.diagonal.sol}.
\end{proof}

In the applications considered in this paper, the above results on the penalty parameters are used as follows. Firstly, we apply the \pdglasso\ procedure with $\lambda_{2}=0$ to a sequence of $m$ values of $\lambda_{1}$ equally spaced, on the log-scale, between $\frac{\lambdadiag}{m}$ and $\lambdadiag$, and then identify the optimal value of $\lambda_{1}$, among the $m$ considered. Next, with the latter value of $\lambda_{1}$, we apply again the \pdglasso\  to a sequence of $m$ values of $\lambda_{2}$ equally spaced, on the log-scale, between $\frac{\lambdasym}{m}$ and $\lambdasym$. The optimal solution is then chosen among the latter $m+1$ solutions identified, including that with $\lambda_{2}=0$. As a goodness of fit measure we use the extended Bayesian Information Criterion (eBIC) of \citet{foygel2010extended},
\begin{equation}\label{eBIC}
\mbox{eBIC}= -n\,l(\widehat{\Theta}^{\mle})+\log(n)\,d + 4\,d\,\gamma\,\log(p),
\end{equation}
where we either set $\gamma=0$, thereby reverting to the classical BIC, or $\gamma=0.5$, as suggested by \citet{foygel2010extended}, when a higher level of sparsity is required. Furthermore, $d$ is the number of parameters, computed by subtracting from the number of parameters of the saturated model, that is
$p(p+1)/2$, both the number of zero constraints and the number of equality constraints characterizing the relevant \PDRCON\ model. It is also worth recalling that $\widehat{\Theta}^{\mle}$ in (\ref{eBIC}) is the (unpenalized) Maximum Likelihood Estimate (MLE) of $\Theta$ \citep{foygel2010extended}, that we compute by means of the same algorithm implemented to solve the \pdglasso\ problem; see Section~\ref{SEC:appendix.algo} for details. This model selection procedure is implemented in the function \texttt{pdRCON.fit} of the \texttt{pdglasso} package.

The application considered in Section~\ref{SEC:application} concerns the identification of gene networks from gene expression data, whereas \citet{ranciati2021fused} considered the identification of brain networks from fMRI data. These are are both relevant areas of application where variables are measured on the same unit. We remark that this is an important issue because lasso methods are not invariant to scalar multiplication of the variables \citep[][p.~9]{hastie2015statistical}. Thus, when the variable units are not the same, it is common practice to standardized the data before applying the traditional graphical lasso \citep{carter2024partial}. However, as noticed by \citet[][Section~3.4]{hojsgaard2008graphical}, \RCON\ models are not invariant under rescaling, in the sense that standardization will not preserve the original structure of colour classes. This means that, although the application of the \pdglasso\ to the sample correlation matrix may still represent a useful way to choose parsimonious \GGM{s}, care needs to be taken in the interpretation of the resulting symmetries. \citet{roverato2024exploration} considered the implementation of greedy search procedures on the space of \PDRCON\ models, for which the standardization of the variables is not needed. However, they are computationally demanding and, thus, their use is limited to problems of small dimension. Furthermore, unlike graphical lasso methods, their application in the high-dimensional setting, with $p$ larger than $n$, is not straighforward.

\section{\PDRCON\ submodel classes}\label{SEC:submodel.classes}
When fitting a \PDRCON\ model to a set of data it can be useful to restrict the analysis to one of some relevant submodel classes. This may be motivated by a number of different reasons, including the interpretability of the selected model or the need to keep the model dimension low, for example when $n$ is small relative to $p$. The function \texttt{admm.pdglasso} of the \texttt{\pdglasso} package implements a more flexible version of the fused penalty (\ref{EQN:fused-penalty}), given by
\begin{align}\label{EQN:piecewise.fused-penalty}
\mathcal{Q}_{\lambda_{2}}(\Theta)=
\lambda_{2}^{(V)}\norm{\diag(\Theta_{LL})-\diag(\Theta_{RR})}_{1}
+
\lambda_{2}^{(I)}\norm{\Theta_{LL}^{\ast}-\Theta_{RR}^{\ast}}_{1}
+
\lambda_{2}^{(A)}
\norm{\Theta_{LR}-\Theta_{RL}}_{1}
\end{align}
where $\Theta_{LL}^{\ast}=\Theta_{LL}-\diag(\Theta_{LL})$ and, similarly for $\Theta_{RR}^{\ast}$; note that $\diag(\Theta)$ refers to a diagonal matrix having the same diagonal as $\Theta$. In this way, there is one regularization parameter associated with every type of parametric symmetry and, more specifically, $\lambda_{2}^{(V)}$ is associated with vertex symmetry, $\lambda_{2}^{(I)}$ with inside-block symmetry and $\lambda_{2}^{(A)}$ with across-block symmetry. For a given value of $\lambda_{2}>0$, each of the three $\lambda_{2}^{(\cdot)}$ parameters can take any of the values in the set $\{0, \lambda_{2},\texttt{Inf}\}$. For instance, the user can impose (i) no constraints involving vertex symmetries by setting $\lambda_2^{(V)}=0$, (ii) the amount of vertex symmetry regularisation implied by $\lambda^{(V)}_{2}=\lambda_{2}$ and (iii) full vertex symmetry by setting $\lambda_2^{(V)}=\texttt{Inf}$. The same can be done independently for  $\lambda_{2}^{(I)}$ and $\lambda_{2}^{(A)}$ thereby allowing the user to select a model within one of $|\{0, \lambda_{2},\texttt{Inf}\}|^{3}=9$  different \PDRCON\ submodel classes, and in the following we look at some of these submodels in more detail; see also Appendix~\ref{SEC:appendix.examples} for some examples with $|L|=|R|=3$.

The restrictions used to define an \RCON\ model are linear in the concentration matrix so that the log-likelihood function is concave \citep{hojsgaard2007inference,hojsgaard2008graphical,gehrmann2011lattices} and, thus, the \pdglasso\ in (\ref{EQN:pdglasso}) is a convex optimization problem. It is therefore computationally convenient to define symmetries in terms of equality of concentration values, although the interpretation of such constraints may not be straightforward.  \citet[][Section~8]{hojsgaard2008graphical} remarked that the comparison of concentration values is meaningful only
when variables are measured on comparable scales, and we note that this is a condition that is naturally satisfied in our framework because, firstly, for every $i\in L$ the corresponding variables $Y_{i}$  and  $Y_{i^{\prime}}$ measure a common feature and, secondly, \PDRCON\ models are defined in such a way that comparisons only involve comparable concentration values, associated with pairs of corresponding variables. Hence, for instance,
the equality $\theta_{ii}=\theta_{i^{\prime}i^{\prime}}$, for $i\in L$, is equivalent to $\sigma^{2}_{ii\mid V\setminus \{i\}}=\sigma^{2}_{i^{\prime}i^{\prime}\mid V\setminus \{i^{\prime}\}}$, that is, that the corresponding variables $Y_{i}$ and $Y_{i^{\prime}}$ have the same partial variance; recall that
$\sigma^{2}_{ii\mid V\setminus \{i\}}=\theta_{ii}^{-1}$ \citep[see, e.g.][Section~5.1.3]{lauritzen1996graphical}. Furthermore, if for instance, for $j\in L$ with $j\neq i$ it holds that both $\theta_{ii}=\theta_{i^{\prime}i^{\prime}}$ and $\theta_{jj}=\theta_{j^{\prime}j^{\prime}}$, then the constraint $\theta_{ij}=\theta_{i^{\prime}j^{\prime}}$ implies the equality of the two partial correlations $\rho_{ij\mid V\setminus \{i,j\}}=\rho_{i^{\prime}j^{\prime}\mid V\setminus \{i^{\prime},j^{\prime}\}}$;
see Appendix~\ref{SEC:full.symmetry} for details.
It is therefore of interest to consider the subfamily of \PDRCON\ models  satisfying full vertex symmetry, that is  $\theta_{ii}=\theta_{i^{\prime}i^{\prime}}$ for every $i\in L$, which are easy to interpret because every equality constraint between off-diagonal concentrations implies the equality of the corresponding partial correlations. In order to restrict the analysis within the family of \PDRCON\ models satisfying full vertex symmetry is sufficient to set $\lambda_{2}^{(V)}=\texttt{Inf}$, but we remark that the equality of partial variances is a strong assumption that should be properly verified.

It may be the case that the substantive research hypothesis underlying the analysis only concerns the comparison of the association structure of the first group with that of the second. More formally, the focus may be on the comparison of $\Theta_{LL}$ and $\Theta_{RR}$, which are adapted to $\G_{L}$ and $\G_{R}$, respectively, whereas the cross-group association $\Theta_{LR}$ is regarded as a nuisance parameter. In this case, across-block symmetries are of no interest and if the sample size is large it may make sense to set $\lambda_{2}^{(A)}=0$. On the other hand, for smaller sample sizes it is also possible to set $\lambda_{2}^{(A)}=\texttt{Inf}$ so as to impose full across-block symmetry with the aim to reduce model dimensionality.

Finally, the relevant research question may require the identification of a model that is fully symmetric, in the sense that there are no differences between the two groups. A fully symmetric model can be obtained by setting $\lambda_{2}^{(V)}=\lambda_{2}^{(I)}=\lambda_{2}^{(A)}=\texttt{Inf}$ and may be used, for instance, as a benchmark for the comparison with an arbitrarily selected model. Interestingly, fully symmetric models belong to the family of \RCON\ models satisfying permutation symmetry, and we refer to Appendix~\ref{SEC:full.symmetry} for a more detailed account on the properties of this submodel class.

\section{Implementation via ADMM algorithm}\label{SEC:appendix.algo}
%
Following \citet{danaher2014joint} and \citet{ranciati2021fused} we solve the \pdglasso\ optimization problem (\ref{EQN:pdglasso}) using an alternating directions method of multipliers (ADMM) algorithm; see \citet[][Section~6.5]{boyd2011distributed}. This is obtained by splitting the procedure into two nested optimization problems both solved by a specific ADMM algorithm, and where the inner ADMM algorithm is written in a form that makes use of the results for the fused lasso signal approximator given in \citet{friedman2007pathwise}. Our ADMM algorithm is implemented in the function \texttt{admm.pdglasso} of the \texttt{pdglasso} package. The latter extends the algorithm of  \citet{ranciati2021fused} to include across-block symmetries but, in fact, it has been designed to solve the more general problem involving the penalty (\ref{EQN:piecewise.fused-penalty}), and thus it makes it possible for the user to optimize (\ref{EQN:pdglasso}) either over the entire  family of \PDRCON\ models or over one of the subfamilies of \PDRCON\ models described in Section~\ref{SEC:submodel.classes}. Furthermore, the function \texttt{pdRCON.mle} of the same package exploits the ADMM algorithm for the computation of the MLEs. This is achieved by allowing different values of the penalty parameters for the different zero and equality constraints, and then setting the penalty term to a sufficiently large value for the constraints characterizing the model and to zero otherwise.

In the following, we provide a detailed description of the algorithm, and refer to \citet{boyd2011distributed} for its convergence properties.
The \pdglasso\ solution is obtained by looping, over iterations $l=1,2,\dots$, the following three steps \citep[see][equations (3.5) to (3.6)]{boyd2011distributed}:
\begin{enumerate}[label=(\arabic*)]
\item $\displaystyle \Theta^{l+1}:=\argmin_{\Theta}\left( -\log\det(\Theta)+\tr(S\Theta)+\dfrac{\rho_1}{2}\norm{\Theta-Z^l+U^l}^2_{\text{F}} \right);$
\item $\displaystyle Z^{l+1}:=\underset{Z}{\argmin}\biggl(\mathcal{P}_{\lambda_{1}}(Z) +  \mathcal{Q}_{\lambda_{2}}(Z)  +\dfrac{\rho_1}{2}\norm{\Theta^{l+1}-Z+U^l}^2_{\text{F}}\biggl);$
\item $\displaystyle U^{l+1}:=U^l+\Theta^{l+1}-Z^{l+1}$,
\end{enumerate}
where $\norm{\cdot}_{\text{F}}$ denotes the Frobenius norm, $\rho_1>0$ is the step size and $U$ and $Z$ are $p \times p$ matrices initialized with all entries equal to zero. The solution to step (1) can be obtained in analytic form as detailed in \citet[][Section~3.1.1]{boyd2011distributed}; see also \citet{ranciati2021fused}.

More specific to our implementation is the solution of step (2). Consider a matrix $Q$ whose columns and rows are indexed by $V=L\cup R$, and let $\vech(\cdot)$ and $\vd(\cdot)$ denote the half-vectorization and the diagonal extraction operator, respectively. Hence, we define the vector $\myvec(Q)$ as,
\begin{align*}
\myvec(Q)^{\top}=
\left[\arraycolsep=0.2pt
\begin{array}{ccccccc}
\vd(Q_{LL})^{\top} &
\vd(Q_{RR})^{\top} &
\vech(Q_{LL})^{\top} &
\vech(Q_{RR})^{\top} &
\vech(Q_{LR})^{\top} &
\vech(Q_{RL})^{\top} &
\vd(Q_{LR})^{\top}
\end{array}
\right],
\end{align*}
and then we set $z^{l} = \myvec(Z^{l})$ and $b^{l} = \myvec(\Theta^{l}) + \myvec(U^{l})$. All equality constraints are encoded in a matrix $F$, made up of the following three row blocks: $[I_{q},\, -I_{q},\, O_{q,4s+q}]$, $[O_{s,2q},\, I_{s},\,  -I_{s},\, O_{s,2s+q}]$ and $[O_{s,2q+2s},\, I_{s},\, -I_{s},\, O_{qq}]$,
where $I$ and $O$ are the identity and zero matrix, respectively, and $s=q(q-1)/2$. Thus, step (2) can be stated as the following
self-standing optimization problem,
\begin{equation*}
\argmin_{z} \biggl(\dfrac{1}{2}\norm{z-b}^2_{2}+\lambda^{\prime}_{1}\!\norm{z}_1+\norm{F(\boldsymbol{\lambda}^{\prime}_{2} \circ z)}_1\biggl),
\end{equation*}
where $\circ$ denotes the element-wise product, $\lambda^{\prime}_{1}=\lambda_1/\rho_1,$ ${\boldsymbol{\lambda}^{\prime}}^{\top}_{2}=\left[  \dfrac{\lambda_{2}^{(V)}}{\rho_1} {\boldsymbol{1}_{2q}}^{\top}  \,\,\,\,\,  \dfrac{\lambda_{2}^{(I)}}{\rho_1}  {\boldsymbol{1}_{2s}}^{\top}\,\,\,\,\, \dfrac{\lambda_{2}^{(A)}}{\rho_1}   {\boldsymbol{1}_{2s+q}}^{\top}  \right],$ and $\boldsymbol{1}_q$ the unit vector of length $q$.

Due to \citet[Lemma A.1]{friedman2007pathwise}, we can firstly solve
\begin{equation*}
\argmin_{z} \biggl(\dfrac{1}{2}\norm{z-b}^2_{2}+\norm{F(\boldsymbol{\lambda}^{\prime}_{2} \circ z)}_1\biggl),
\end{equation*}
with $\lambda^{\prime}_{1}=0$, and then obtain the solution for the case $\lambda^{\prime}_{1}>0$ via a soft-thresholding operation. This is a generalized lasso problem \citep{tibshirani2011solution}, which we solve with the following inner ADMM procedure \citep[Section~6.4.1]{boyd2011distributed},
\begin{enumerate}[label=(\roman*)]
\item $z^{m+1}:= \bigl(I+\rho_2F^{\top}\!F\bigl)^{-1}\left\{b+\rho_2F^{\top}\!(v^m-t^m)\right\}$;
\item $v^{m+1}:=\mathcal{S}_{\boldsymbol{\lambda}^{\prime}_2/\rho_2}(Fz^{m+1}+t^m)$;
\item $t^{m+1}:=t^m+Fz^{m+1}-v^{m+1}$.
\end{enumerate}
The vectors $v$ and $t$ are initialized with zero entries, the scalar $\rho_2 >0$ refers to the step size of the inner ADMM and $\mathcal{S}_{\boldsymbol{\lambda}^{\prime}_2/\rho_2}(\cdot)$ represents the soft-thresholding operator, such that any input less than $\boldsymbol{\lambda}^{\prime}_2/\rho_2$ (in absolute value) is set to zero, otherwise it gets shrunk by that threshold.
The stopping rule of the algorithm is based on the primal and dual residual \citep[Section~3.3.1]{boyd2011distributed}, computed at each iteration, which are compared to a given numerical precision threshold.

Finally, we remark that, in the computer implementation of the above ADMM algorithm, efficiency has been achieved by making the
step sizes $\rho_{1}$ and $\rho_{2}$ adaptive and, furthermore, both by properly reducing the dimension of the matrix $F$ according to the submodel class of interest, and by encoding the results of the products involving the matrix $F$ in the form of less expensive vectorized computations, explicitly based on the (sparse) structure of such matrix. In a call to the function \texttt{pdRCON.fit} on a grid of 20 values for each penalty parameter, the ADMM algorithm is run 80 times, that is 40 times on a range of different penalty values and 40 times for the computation the MLEs for the eBIC. On a recent hardware (CPU 1.8GHz Intel Core i7 quad-core, memory 16 GB 2133 MHz LPDDR3), the individual cost for an ADMM execution is approximately 1 second. These computations refer to execution times for the simulation study described in the Section \ref{SEC:simulations}.

\section{Simulation study}\label{SEC:simulations}
\PDRCON\ models are \GGM{s} with additional equality constraints, which we call parametric symmetries. It follows that a \GGM\ is a \PDRCON\ model with no parametric symmetry, and it is therefore of interest to investigate the behaviour of the \pdglasso\ method when the data are generated from a \GGM\ and, similarly, the behaviour of the glasso method when parametric symmetries are indeed present.

In the following, we compare the behaviour of glasso and \pdglasso\ by means of simulated data, and consider three different scenarios: (i) \GGM{s} with no parametric symmetry, (ii) \PDRCON\ models where 50\% of the parameters are involved in a parametric symmetry and the remaining 50\% are unconstrained and, (iii) fully symmetric \PDRCON\ models.
For each scenario we set $p=100$ and randomly generated 10 models with graph density equal to $0.20$, where this density is computed as the number of present edges over the number of edges of the complete graph. A Gaussian distribution, identified by its covariance matrix, was obtained for every model and from each of the resulting 30 distributions we randomly generated $7$ datasets of size $n=\{100, 150, 200, 300, 500, 1000, 1500\}$. To each of the 210 resulting datasets, we applied the  procedure described in Section~\ref{SEC:pdglasso.solutions} with grid length $m=20$ and $\gamma=0$. All the computations were carried out by using the \textsf{R} package \texttt{pdglasso}, and we now describe the procedure used to randomly generate the covariance matrices characterizing the Gaussian distributions used in the simulations. The procedure implemented in the function  \texttt{GGM.simulate} to obtain a covariance matrix $\Sigma^{(\ast)}$ relative to a \GGM\ is as follows. Firstly, a $p\times p$ positive definite matrix $S^{\ast}$ is randomly generated from a Wishart distribution with matrix parameter equal to the identity matrix. Next, $(S^{\ast})^{-1}$ is computed and an undirected graph $\G$, with the required density, is obtained in such a way that the edges of $\G$ correspond to the largest absolute values of the off-diagonal entries of $(S^{\ast})^{-1}$. Finally, a positive definite matrix $\Sigma^{(\ast)}$, that identifies a distribution from the \GGM\ represented by $\G$, is obtained by applying to $S^{\ast}$ the procedure for the computation of the MLEs. The same idea is exploited by the function \texttt{pdRCON.simulate} to add to $\Sigma^{(\ast)}$ the required amount of symmetries so as to obtain a covariance matrix from a \PDRCON\ model. For information about the relevant material to reproduce the simulations see the Section ``Code and data availability'' of this paper.

For every scenario and sample size, the average behaviour across the 10 selected models of  \pdglasso\ was compared with that of  glasso. Because glasso is not designed to identify parametric symmetries, we compared the two methods with respect to their performance in learning the graph structure and the inverse covariance matrix $\Theta$.

The performance of the procedures in structure learning was measured through the Positive Predicted Value (PPV) and the True Positive Rate (TPR). These are given in Figure~\ref{FIG:simulations.01} where also a measure of parsimony is provided, by the number of parameters of the selected models.
The panels composing the first column of Figure~\ref{FIG:simulations.01} show that, for the no symmetry case where the true model is a GGM, the \pdglasso\ method selected less parsimonious models, thereby resulting in a worse performance, compared to glasso, with respect to PPV but an improvement in terms of TPR; these differences between the two methods are more noticeable for larger sample sizes. Plots in the second column show that in the 50\% symmetry case the two methods performed similarly. For the full symmetry case, third column, \pdglasso\ selected more parsimonious models, with a better performance in terms of PPV and producing TPR values comparable to those of glasso.

\begin{figure}[tbp]
\centering
\includegraphics[scale=0.38]{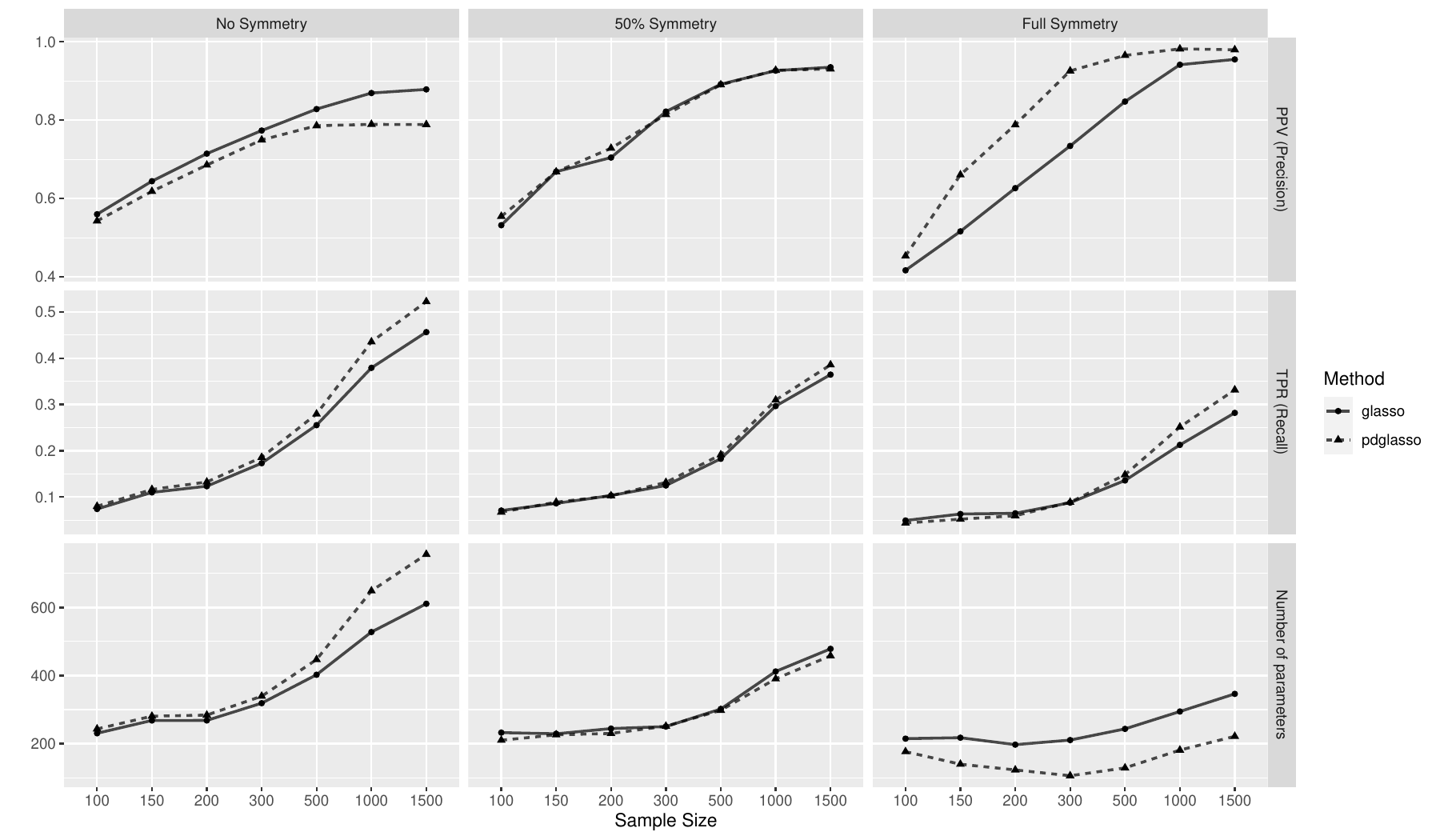}
\caption{\small Comparison of the \pdglasso\ and the glasso methods for the three scenarios: ``no symmetry'', ``50\% symmetry'' and ``full symmetry''. Panels in the first row give the Positive Predicted Value (PPV) also called precision, those in the second row the True Positive Rate (TPR) also called recall, whereas the third row gives the number of parameters. Every value is the average of the corresponding quantities computed on the 10 selected models for every scenario and sample size.}
\label{FIG:simulations.01}
\end{figure}
\begin{figure}[tbp]
\centering
\includegraphics[scale=0.38]{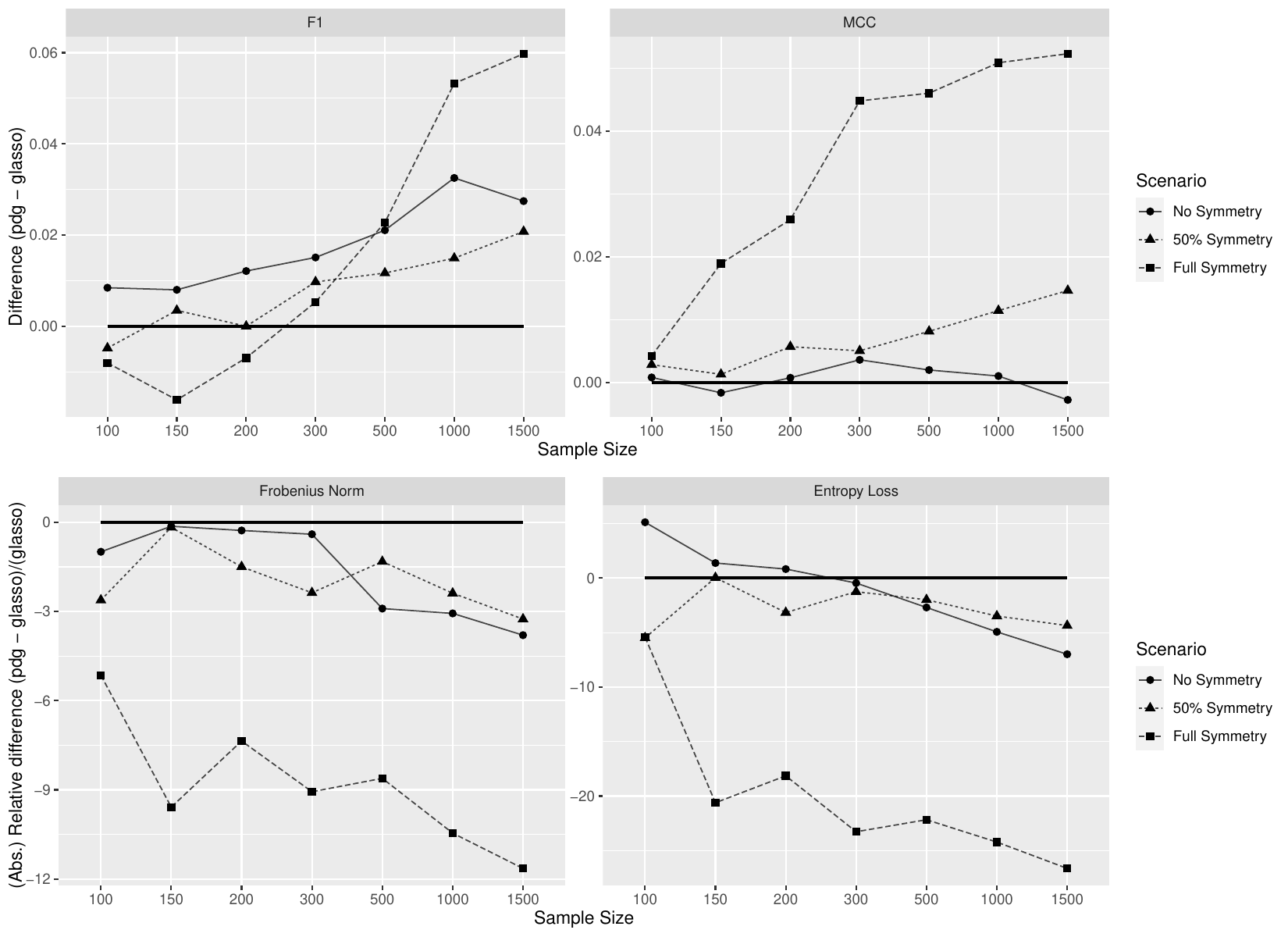}
\caption{\small Comparison of the \pdglasso\ and the glasso methods for the three scenarios: ``no symmetry'', ``50\% symmetry'' and ``full symmetry''. Top panels: differences of the F1 score values and of the Matthews correlation coefficient (MCC) values, with positive value denoting a better performance of \pdglasso. Bottom panels: relative differences of the Frobenius norm loss values for the inverse covariance and of the entropy loss values, with negative value denoting a better performance of \pdglasso. Every value is the average of the corresponding quantities computed on the 10 selected models for every scenario and sample size.}
\label{FIG:simulations.02}
\end{figure}

As overall measures of performance in structure learning we considered the harmonic mean of PPV and TPR, called the F1 score, and the Matthews correlation coefficient (MCC). Furthermore, we assessed the performance with respect to the estimate of the inverse covariance by means of the Frobenius norm loss and of the entropy loss \citep{dey1985estimation}. 
More specifically the top panels of Figure~\ref{FIG:simulations.02} compare the two method by means of the differences of the F1 scores and the MCC values, respectively, with positive values in favour of \pdglasso\, whereas the bottom panels give the relative differences of the Frobenius norm loss and of the entropy loss, with negative values in favor of \pdglasso. 
Figure~\ref{FIG:simulations.02} shows that, although there are some cases where glasso performs moderately better than \pdglasso\, especially for small sample sizes and in the no symmetry case, there is a clear evidence that, as expected, the \pdglasso\ method provided better results in the models presenting parametric symmetry. Additionally, in the case where the data
come from a GGM the \pdglasso\ method provided results comparable with those of glasso.

\section{Analysis of breast cancer gene expression data}\label{SEC:application}
We illustrate the use of our proposed method in a gene expression paired data problem concerning breast cancer. The samples refer to $n=114$ individuals with both tumor and healthy adjacent tissue information, hence the paired data nature of the problem. A curated set of $q=89$ genes of the Hedgehog Pathway, known for their involvement in breast cancer \citep{song2014pathway,kubo2004hedgehog}, is extracted from MSigDB Collections \citep{subramanian2005gene}. For each individual, we consider gene-level transcription estimates, that is $\text{log}_2(Y+1)$ transformed normalized counts, across the same set of selected genes in both tissues ($p=89\times2=178$), which lead to the $114 \times 178$ final data matrix for the analysis.

We focus on two pdRCON submodel classes. Firstly, we consider a \PDRCON\ model (labelled as \texttt{fV}) where $\lambda^{(V)}_{2}=\texttt{Inf}$ and  $\lambda^{(I)}_{2} = \lambda^{(A)}_{2}=\lambda_{2}$, so as to force full vertex symmetry, and penalize with magnitude $\lambda_{2}$ the remaining blocks. The second model considered  (labelled as \texttt{fVIA}) is a fully symmetric model, which is obtained by setting $\lambda^{(V)}_{2}=\lambda^{(I)}_{2} = \lambda^{(A)}_{2}=\texttt{Inf}$. More details on these submodel classes can be found in  Section~\ref{SEC:submodel.classes} and Appendix~\ref{SEC:full.symmetry}. For the selection of the optimal penalty values we apply the procedure described in Section~\ref{SEC:pdglasso.solutions}, with grid length $m=20$, and where the eBIC is applied with $\gamma=0.5$, because domain experts are interest in recovering a sparse graph with the most relevant connections. The selected models, denoted as \texttt{mod$_{\text{fV}}$} and \texttt{mod$_{\text{fVIA}}$}, respectively, are
visually depicted in  Figures~\ref{FIG:appendix.plots.fV} and \ref{FIG:appendix.plots.fVIA} of Appendix~\ref{SEC:appendix.plots}. They are very sparse and, for this reason, in the following we will only focus on the symmetries, either structural and parametric, that involve edges which are present in the graphs. Accordingly, to ease the reading, we will simply talk of symmetries without recalling that we are only considering present edges.

\begin{table}
\caption{\label{TAB:G.summary}Summary statistics on the structure of the models identified by the application of \pdglasso: \texttt{mod$_{\text{fVIA}}$} (last column); \texttt{mod$_{\text{fV}}$} and its three submodels considered (second to fifth column); \texttt{sub1$_{\text{fV}}$}: across-block edges set to zero; \texttt{sub2$_{\text{fV}}$}: inside-block structural \textbf{a}symmetries removed; \texttt{sub3$_{\text{fV}}$}: inside-block structural symmetries converted to parametric symmetries. All models have full vertex symmetry.}
\center\begin{tabular}{l|cccc|c}
\hline
 & \multicolumn{4}{|c|}{Reference model and submodels}  & \\
Summary & \texttt{mod$_{\text{fV}}$} & \texttt{sub1$_{\text{fV}}$}  & \texttt{sub2$_{\text{fV}}$} &  \texttt{sub3$_{\text{fV}}$} & \texttt{fVIA} \\
\hline
Total edges & 117 & 110 & 49 & 117 & 300\\
Graph density & 0.74 \% & 0.70 \% & 0.31 \% & 0.74 \% &1.90 \% \\
\hline
\multicolumn{6}{c}{\emph{inside block}} \\
\hline
Total edges & 110 & 110 & 42 & 110 & 286 \\
Structural (non-zero) symmetric edges & 18 & 18 & 18 & 0 &0 \\
Parametric (non-zero) symmetric edges & 24 & 24 & 24 & 42 & 286\\
\hline
\multicolumn{6}{c}{\emph{across block}} \\
\hline
Total edges & 7 & 0 & 7 & 7 & 14\\
Structural (non-zero) symmetric edges & 0 & 0 & 0 & 0 & 0\\
Parametric (non-zero) symmetric edges & 0 & 0 & 0 & 0 & 10\\
\hline
\end{tabular}
\end{table}

An overview of the structure of the models involved in this analysis is provided in Table~\ref{TAB:G.summary}. The graph associated to model \texttt{mod$_{\text{fV}}$} has an overall density lower than 1\%, with inside-block density larger than the across-block one. In particular, the graph has only 7 edges connecting genes across the two different groups (tumor vs healthy), whereas the two group specific graphs are similar in terms of density, with 56 edges between variables in the tumor group and 54 in the healthy group. The most connected gene, both inside and across blocks, is the one named LRP2. More specifically, in the tumor group 16 out of the 56 edges ($28.6\%$) connect LRP2 to other inside-block genes, whereas in the healthy group 37 out of the 54 edges ($68.5\%$) involve LRP2. Furthermore, most of the across-block edges, i.e. 5 out of the 7 ($71.4\%$), involve LRP2 and,  finally, 12 of the 24 coloured edge involve LRP2, that is $50\%$ of the identified parametric symmetries concern this gene.

We turn now to the fully symmetric model \texttt{mod$_{\text{fVIA}}$} and note that, although in this case the identified graph has a larger number of edges, with a density that is almost three times greater than that of the graph associated with \texttt{mod$_{\text{fV}}$}, the number of parameters of \texttt{mod$_{\text{fVIA}}$} is only $1.3$ times larger than that of \texttt{mod$_{\text{fV}}$}, which shows how the identification of parametric symmetries can provide an effective way to increase  parsimony in the selection of a \GGM.

\begin{table}
\caption{\label{TAB:LRT} Summary of Likelihood Ratio Tests (LRTs) performed on each submodel against the reference \texttt{mod$_{\text{fV}}$}.}
\center\begin{tabular}{l|cccc}
\hline
$H_0:$ submodel is preferable & \texttt{mod$_{\text{fV}}$} & \texttt{sub1$_{\text{fV}}$}  & \texttt{sub2$_{\text{fV}}$} &  \texttt{sub3$_{\text{fV}}$} \\ 
\hline
Number of parameters & 194 & 187 & 126 & 185  \\ 
Deviance & 12242.22 & 12285.79  & 14218.92 & 12393.96 \\
LRT value & - & 43.56 & 1976.69 & 151.73 \\
Degrees of freedom (df) & - & 7  & 68 & 9 \\
Critical $\chi_{df,\alpha}$ quantile at $\alpha=0.95$ & - & 14.07 & 88.25 & 16.92 \\
\hline
\end{tabular}
\end{table}

We also compare models \texttt{mod$_{\text{fV}}$} with three of its submodels: (i) \texttt{sub1$_{\text{fV}}$}, where across-block edges are removed; (ii) \texttt{sub2$_{\text{fV}}$}, where inside-block structural \textbf{a}symmetries are removed; (iii) \texttt{sub3$_{\text{fV}}$}, where all inside-block structural symmetries are turned into parametric symmetries. The three submodels are not nested within one another but they all are with respect to \texttt{mod$_{\text{fV}}$}.

We first consider model \texttt{sub1$_{\text{fV}}$}. Unlike the case of independent samples, in paired data problems a key issue is represented by the cross-graph structure that encodes the across-group dependence. It is interesting to note that there are potentially $q^{2}=7921$ across-graph edges, but the selected model \texttt{mod$_{\text{fV}}$} shows a very sparse across-graph association structure with as few as 7 edges, most of which involving the gene LRP2. It is therefore of interest to quantify the significance of the across-graph structure by comparing the model \texttt{mod$_{\text{fV}}$} with its submodel \texttt{sub1$_{\text{fV}}$} obtained by assuming that the group specific subgraphs are disconnected, that is $\Theta_{LR}=O_{q}$. This comparison is carried out by means of the  Likelihood Ratio Test (LRT) at a significance level of $\alpha=0.05$, given in Table~\ref{TAB:LRT}, that shows  that the model  \texttt{sub1$_{\text{fV}}$} seems not to provide an adequate fit to the data, and thus we can conclude that the across-groups structure cannot be ignored, and should thus be retained.

In the comparison of the two groups, differences are represented by \textbf{a}symmetries, and this motivates the comparison of the selected model
\texttt{mod$_{\text{fV}}$} with its submodels \texttt{sub2$_{\text{fV}}$}  and  \texttt{sub3$_{\text{fV}}$}. As shown but juxtaposing the second and fourth column of Table~\ref{TAB:G.summary}, among the 110 inside-block edges of the reference model 68 are \textbf{a}symmetric edges and thus not involved in any structural or parametric symmetry and model \texttt{sub2$_{\text{fV}}$} is obtained by removing such 68 edges. Hence, comparing \texttt{mod$_{\text{fV}}$} with \texttt{sub2$_{\text{fV}}$} amounts to test the presence of structural differences between the group specific subgraphs. Finally, comparing \texttt{mod$_{\text{fV}}$} with \texttt{sub3$_{\text{fV}}$} amounts to test if the model can be further simplified by assuming that every structural symmetry is also parametric. As shown in Table~\ref{TAB:LRT} the empirical evidence is in favor of the hypotheses that structural asymmetries are present and that not all structural symmetries are also parametric.

\section{Conclusions}\label{SEC:conclusions}
We have considered the problem of joint learning of \GGM{s} for paired data, in an approach that implements a fused graphical lasso penalty to identify a model within a suited family of coloured graphical models. We have addressed a number of issues with the aim of providing the results and tools required for an immediate and knowledgeable application of the method. The latter include an ADMM algorithm for the optimization of the penalized likelihood, some results required when computing a path of lasso solutions, the description of the features of some model subfamilies of specific interest and, finally, an \textsf{R} package where all the methods have been implemented.

We deem that one appealing feature of our approach is the interpretation of the selected model. When the number of variables is large, it is not straightforward to visualize and interpret a \GGM, and this task is even more challenging when the analysis involves two or more, possibly dependent, networks. The application of \PDRCON\ models is restricted to the case where two groups are considered, but it provides a transparent representation of the across group association. Furthermore, as shown in the application of Section~\ref{SEC:application}, one can meaningfully summarize the relationship between groups by focusing on the amount of both structural and parametric symmetry/asymmetry. The scope of the analysis may also justify the restriction to specific submodels such as the model with full vertex symmetry of the fully symmetric model.

\section*{Code and data availability}
All implemented scripts and functions used throughout the paper (and more) are available at the following GitHub repository: \url{https://github.com/savranciati/pdglasso}. For reproducibility reasons, both original data (downloaded from \url{https://xenabrowser.net/}) and sample covariance matrices are available as part of the \textsf{R} package \texttt{pdglasso}, together with their metadata information. Scripts and datasets pertaining to the simulation study are available online as Supplementary Material.

\section*{Acknowledgments}
We would like to thank two anonymous referees for their valuable comments, and Alberto Tonolo, Ugo Ala and Sara Ferrini for useful discussion and the support provided in pre-processing the cancer genomics data.

\section*{Funding}
Financial support was provided by the European Union - Next Generation EU, Mission 4 Component 2 - CUP~C53D23002580006
(MUR-PRIN grant~2022SMNNKY)

\section*{Disclosure statement}
The authors report there are no competing interests to declare.

\bibliographystyle{chicago}
\bibliography{biblio-pdglasso}


\appendix
\newpage
\section*{Appendices}

\section{Fully symmetric \PDRCON\ models}\label{SEC:full.symmetry}
\newcommand{\A}{A.}
\renewcommand{\theequation}{\A\arabic{equation}}
\renewcommand{\thetable}{\A\arabic{table}}
\renewcommand{\thefigure}{\A\arabic{figure}}
\setcounter{figure}{0}
\setcounter{equation}{0}

One of the  main motivations for the analysis of paired data is the need to identify similarities, i.e. symmetries, and differences between the two groups. Thus, in this framework, a central role is played by the subclass of \PDRCON\ models that show no differences between the two groups. It is of theoretical interest to understand the properties of this model class and, furthermore, in applications it may be useful to estimate a model within this class because it could be regarded as a benchmark for the comparison with  an arbitrarily selected model, and the quantification of the amount of asymmetry it shows. We call these models \emph{fully symmetric} because, as shown below, their parametric symmetries are not restricted to the equalities of concentration coefficients but extend to  partial correlation coefficients, variances, covariances and correlation coefficients.  Theorem~\ref{THM:hat.fully.sym} provides a result concerning the value of $\lambda_{2}$ such that the solution to the $\pdglasso$ is a fully symmetric model, whereas an instance of a fully symmetric \PDRCON\ model is given in Example~\ref{EXA:app.4} of Appendix~\ref{SEC:appendix.examples}.

As far as the interpretation of the model is of concern, a useful property of fully symmetric \PDRCON\ models is that every inside and across-block parametric equality constraint, for instance $\theta_{ij}=\theta_{i^{\prime}j^{\prime}}$,
also implies that the corresponding partial correlations have the same values, i.e. $\rho_{ij\mid V\setminus\{i,j\}}=\rho_{i^{\prime}j^{\prime}\mid V\setminus\{i^{\prime},j^{\prime}\}}$. This follows immediately from the fact that the partial correlation between $Y_{i}$ and $Y_{j}$ given the remaining variables can be computed from concentrations as $\rho_{ij\mid V\setminus\{i,j\}}=-\theta_{ij}/\sqrt{\theta_{ii}\theta_{jj}}$ \citep[see, e.g.][Section~5.1.3]{lauritzen1996graphical} and in a fully symmetric model it also holds that both $\theta_{ii}=\theta_{i^{\prime}i^{\prime}}$ and $\theta_{jj}=\theta_{j^{\prime}j^{\prime}}$.

An equivalent way to define fully symmetric \PDRCON\ models is by requiring that the random vector $(Y_{L}, Y_{R})^{\top}$ has the same distribution as that of $(Y_{R}, Y_{L})^{\top}$; that is $Y_{L}$ and $Y_{R}$ are exchangeable. This shows that fully symmetric \PDRCON\ models belong to the family of coloured \GGM{s} which satisfy permutation symmetry, as defined in  \citet[][Section~5]{hojsgaard2008graphical}; see also \citet{gehrmann2011lattices} and \citet{graczyk2022model}. More specifically, if we denote by $J$ the permutation matrix that exchanges $Y_{L}$ and $Y_{R}$, i.e. $J\times (Y_{L}, Y_{R})^{\top}=(Y_{R}, Y_{L})^{\top}$, then we have,
\begin{align}\label{EQN:definition-J}
J =
\left(
\begin{array}{cc}
O_{q} & I_{q}\\
I_{q} & O_{q}\\
\end{array}
\right)
\quad\mbox{and thus}\quad
J \Theta J =
\left(
\begin{array}{cc}
\Theta_{RR} & \Theta_{RL}\\
\Theta_{LR} & \Theta_{LL}\\
\end{array}
\right),
\end{align}
where $I_{q}$ and  $O_{q}$ are the $q\times q$ identity  and zero matrices, respectively. Recall that permutation matrices are orthogonal and, furthermore, $J$ is symmetric so that  $J=J^{\top}=J^{-1}$. We can thus say that a \PDRCON\ model is fully symmetric if and only if $J\Theta J=\Theta$, in words, if and only if \emph{$\Theta$ is invariant under the action of $J$}. It is straightforward to see that the latter equality is equivalent to $(J\Theta J)^{-1}=\Theta^{-1}$ and thus to $J\Sigma J=\Sigma$, and this shows that in a fully symmetric model $\Sigma$ is also invariant under the action of $J$. This implies that parametric symmetry holds also for to the entries of the variance matrix, that is $\Sigma_{LL}=\Sigma_{RR}$ and $\Sigma_{LR}=\Sigma_{RL}$ and it straightforward to see that also  the correlation matrix $P$ satisfies the same property, i.e. both  $P_{LL}=P_{RR}$ and $P_{LR}=P_{RL}$.

Finally, we notice that for the family of fully symmetric models, the MLE of $\Theta$ can be obtained by instead considering the matrix $\bar{S}=(S+JSJ)/2$ and finding the MLE of the corresponding \GGM{} without symmetry restrictions \citep[][Section~5.2]{hojsgaard2008graphical}. This MLE can be calculated explicitly when $\G$ is decomposable, or alternatively a standard algorithm, such as the iterative proportional scaling, can be used \citep[][page~146]{lauritzen1996graphical}.

\section{Proofs}\label{SEC:proofs}

\subsection{Proof of Theorem~\ref{THM:hat.fully.sym}}\label{SEC:appendix.proof.main.theorem}
\newcommand{\B}{B.}
\renewcommand{\theequation}{\B\arabic{equation}}
\renewcommand{\thetable}{\B\arabic{table}}
\renewcommand{\thefigure}{\B\arabic{figure}}
\setcounter{equation}{0}

For a positive definite matrix $\Theta$ we consider $J$ in (\ref{EQN:definition-J}) and set the matrix $\bar{\Theta}=\frac{1}{2}(\Theta+J\Theta J)$. Note that
$J^{-1}=J$ so that $J\bar{\Theta}J=\frac{1}{2}(J\Theta J+JJ\Theta JJ)=\frac{1}{2}(J\Theta J+\Theta)=\bar{\Theta}$, that is,
$\bar{\Theta}$ is fully symmetric. It is also worth remarking that $\bar{\Theta}$ is positive definite by construction.

We now show that if $\lambda_{2}$ is chosen to satisfy
the inequality of Theorem~\ref{THM:hat.fully.sym}  then for  any positive definite matrix $\Theta$ it holds that $\mathcal{L}_{\lambda_{1},\lambda_{2}}(\bar{\Theta})\leq \mathcal{L}_{\lambda_{1},\lambda_{2}}(\Theta)$ so that the matrix which minimises $\mathcal{L}_{\lambda_{1},\lambda_{2}}(\cdot)$ in (\ref{EQN:pdglasso}) must be fully symmetric. Concretely, we show that
$\mathcal{L}_{\lambda_{1},\lambda_{2}}(\bar{\Theta})- \mathcal{L}_{\lambda_{1},\lambda_{2}}(\Theta)\leq 0$ and, to this aim, we analyse, in turn, three components of this difference.

We first show that $\{-\log\det(\bar{\Theta})\}-\{-\log\det(\Theta)\}\leq 0$,  which is equivalent to $\log\det(\Theta)\leq \log\det(\bar{\Theta})$ and therefore to $\det(\Theta)\leq \det(\bar{\Theta})$. Because both $\Theta$ and $J\Theta J$ are positive definite the following inequality holds
\citep[see, e.g.,][p.~55 eqn (14)]{lutkepohl1996handbook}
\begin{align*}
\det(\bar{\Theta})^{1/p}
&= \det\left(\frac{1}{2}\Theta+\frac{1}{2}J\Theta J\right)^{1/p}\geq \det\left(\frac{1}{2}\Theta\right)^{1/p} + \det\left(\frac{1}{2}J\Theta J\right)^{1/p}\\
&= 2 \det\left(\frac{1}{2}\Theta\right)^{1/p}= \det(\Theta)^{1/p}
\end{align*}
where we have used the facts that $\det(\Theta)=\det(J\Theta J)$, because $J$ is a permutation matrix, and that $\det\left(\frac{1}{2}\Theta\right)=\frac{1}{2^p}\det(\Theta)$. We have thus shown that $\det(\bar{\Theta})^{1/p}\geq \det(\Theta)^{1/p}$ and this immediately implies that $\det(\bar{\Theta})\geq \det(\Theta)$ as required.

We now show that $+\lambda_1\!\norm{\bar{\Theta}}_1-\lambda_1\!\norm{\Theta}_1\leq 0$, i.e, that $\norm{\bar{\Theta}}_1-\norm{\Theta}_1\leq 0$. We first notice that $\norm{\Theta}_1$ can be computed on the block components
(\ref{EQN:Theta.blocks}) of $\Theta$ as follows,
\begin{align*}
\norm{\Theta}_{1}
&=
 \norm{\Theta_{LL}}_{1}
+\norm{\Theta_{RR}}_{1}
+\norm{\Theta_{LR}}_{1}
+\norm{\Theta_{RL}}_{1}.
\end{align*}
Similarly,
\begin{align*}
\norm{\bar{\Theta}}_{1}
&=\norm{\bar{\Theta}_{LL}}_{1}
+\norm{\bar{\Theta}_{RR}}_{1}
+\norm{\bar{\Theta}_{LR}}_{1}
+\norm{\bar{\Theta}_{RL}}_{1}\\
&=\frac{1}{2}(
 \norm{\Theta_{LL}+\Theta_{RR}}_{1}
+\norm{\Theta_{RR}+\Theta_{LL}}_{1}
+\norm{\Theta_{LR}+\Theta_{RL}}_{1}
+\norm{\Theta_{RL}+\Theta_{LR}}_{1}
)\\
&= \norm{\Theta_{LL}+\Theta_{RR}}_{1}
+\norm{\Theta_{LR}+\Theta_{RL}}_{1}.
\end{align*}
Hence,
\begin{align*}
\norm{\bar{\Theta}}_{1}-\norm{\Theta}_{1}
&= \norm{\Theta_{LL}+\Theta_{RR}}_{1} - (\norm{\Theta_{LL}}_{1}+\norm{\Theta_{RR}}_{1})\\
&+ \norm{\Theta_{LR}+\Theta_{RL}}_{1} - (\norm{\Theta_{LR}}_{1}+\norm{\Theta_{RL}}_{1}),
\end{align*}
and it follows that $\norm{\bar{\Theta}}_{1}-\norm{\Theta}_{1}\leq 0$ because for any pair of real numbers $a$ and $b$ it holds that $|a+b|\leq |a|+|b|$.

In order to complete the proof we have now to show that
\begin{align}\label{EQN:proof3.third.part}
\tr(S\bar{\Theta})
-\tr(S\Theta)-\lambda_2
\left(
\norm{\Theta_{LL}-\Theta_{RR}}_1
+\norm{\Theta_{LR}-\Theta_{RL}}_1
\right)\leq 0.
\end{align}
Note that the term $\lambda_2\left(\norm{\bar{\Theta}_{LL}-\bar{\Theta}_{RR}}_1+\norm{\bar{\Theta}_{LR}-\bar{\Theta}_{RL}}_1\right)$ does not appear in (\ref{EQN:proof3.third.part}) because it is equal to zero by construction. We first consider the difference
\begin{align}
\nonumber
\tr(S\bar{\Theta})-\tr(S\Theta)
\nonumber&=\tr(S\bar{\Theta}-S\Theta)=\tr\{S(\bar{\Theta}-\Theta)\}=\tr\{S(\frac{1}{2}\Theta + \frac{1}{2}J\Theta J- \Theta)\}\\
&=\tr\{\frac{1}{2}S(J\Theta J - \Theta)\}
\label{EQN:proof3.trace.matrix}
=\frac{1}{2}\sum_{i=1}^{p}\sum_{j=1}^{p} s_{ij}(J\Theta J - \Theta)_{ij},
\end{align}
where (\ref{EQN:proof3.trace.matrix}) follows from the fact that both $S$ and $(J\Theta J - \Theta)$ are symmetric matrices.
Hence, if we write (\ref{EQN:proof3.trace.matrix}) by distinguishing the terms corresponding to each of the four blocks of $(J\Theta J - \Theta)$ as in (\ref{EQN:Theta.blocks}) we obtain,
\begin{align}
\nonumber
\tr(S\bar{\Theta})-\tr(S\Theta)
&=\frac{1}{2}\sum_{i=1}^{q}\sum_{j=1}^{q}
\{
s_{ij}(\theta_{i^{\prime}j^{\prime}}-\theta_{ij})
+s_{i^{\prime}j^{\prime}}(\theta_{ij}-\theta_{i^{\prime}j^{\prime}})
+s_{ij^{\prime}}(\theta_{i^{\prime}j}-\theta_{ij^{\prime}})
+s_{i^{\prime}j}(\theta_{ij^{\prime}}-\theta_{i^{\prime}j})
\}\\
\label{EQN:proof3.trace.decomposition}
&=\frac{1}{2}\sum_{i=1}^{q}\sum_{j=1}^{q}
\{
(s_{i^{\prime}j^{\prime}}-s_{ij})(\theta_{ij}-\theta_{i^{\prime}j^{\prime}})
+(s_{i^{\prime}j}-s_{ij^{\prime}})(\theta_{ij^{\prime}}-\theta_{i^{\prime}j})
\}\\
\label{EQN:proof3.trace.inequality}
&\leq
\frac{1}{2}\sum_{i=1}^{q}\sum_{j=1}^{q}
(
|s_{i^{\prime}j^{\prime}}-s_{ij}||\theta_{ij}-\theta_{i^{\prime}j^{\prime}}|
+|s_{i^{\prime}j}-s_{ij^{\prime}}||\theta_{ij^{\prime}}-\theta_{i^{\prime}j}|
).
\end{align}
Furthermore,
\begin{align}\label{EQN:proof3.penalty.decomposition}
\lambda_{2}
\left(
\norm{\Theta_{LL}-\Theta_{RR}}_{1}
+\norm{\Theta_{LR}-\Theta_{RL}}_{1}
\right)
=\sum_{i=1}^{q}\sum_{j=1}^{q}
\lambda_{2}\left(|\theta_{ij}-\theta_{i^{\prime}j^{\prime}}|
+|\theta_{ij^{\prime}}-\theta_{i^{\prime}j}|
\right).
\end{align}
Hence, from (\ref{EQN:proof3.trace.decomposition}),  (\ref{EQN:proof3.penalty.decomposition}) and then   (\ref{EQN:proof3.trace.inequality}) we
obtain
\begin{align*}
&=
 \tr(S\Theta)
-\tr(S\Theta)
-\lambda_2
\left(
\norm{\Theta_{LL}-\Theta_{RR}}_1
+\norm{\Theta_{LR}-\Theta_{RL}}_1
\right)\\
&=\frac{1}{2}\sum_{i=1}^{q}\sum_{j=1}^{q}
\{
 (s_{i^{\prime}j^{\prime}}-s_{ij})(\theta_{ij}-\theta_{i^{\prime}j^{\prime}})
+(s_{i^{\prime}j}-s_{ij^{\prime}})(\theta_{ij^{\prime}}-\theta_{i^{\prime}j})
-\sum_{i=1}^{q}\sum_{j=1}^{q}
\lambda_{2}
\left(
 |\theta_{ij}-\theta_{i^{\prime}j^{\prime}}|
+|\theta_{ij^{\prime}}-\theta_{i^{\prime}j}|
\right)\\
&\leq
\frac{1}{2}\sum_{i=1}^{q}\sum_{j=1}^{q}
(
|s_{i^{\prime}j^{\prime}}-s_{ij}||\theta_{ij}-\theta_{i^{\prime}j^{\prime}}|
+|s_{i^{\prime}j}-s_{ij^{\prime}}||\theta_{ij^{\prime}}-\theta_{i^{\prime}j}|
)
-\sum_{i=1}^{q}\sum_{j=1}^{q}
\lambda_{2}
\left(
 |\theta_{ij}-\theta_{i^{\prime}j^{\prime}}|
+|\theta_{ij^{\prime}}-\theta_{i^{\prime}j}|
\right)\\
&=
\sum_{i=1}^{q}\sum_{j=1}^{q}
\left\{
\left(|s_{i^{\prime}j^{\prime}}-s_{ij}|/2-\lambda_{2}\right)|\theta_{ij}-\theta_{i^{\prime}j^{\prime}}|
+\left(|s_{i^{\prime}j}-s_{ij^{\prime}}|/2-\lambda_{2}\right)|\theta_{ij^{\prime}}-\theta_{i^{\prime}j}|
\right\},
\end{align*}
that establishes (\ref{EQN:proof3.third.part}) because we have assumed that $\lambda_{2}$ satisfies the inequality of Theorem~\ref{THM:hat.fully.sym} so that both
$(|s_{i^{\prime}j^{\prime}}-s_{ij}|/2-\lambda_{2})$ and $(|s_{i^{\prime}j}-s_{ij^{\prime}}|/2-\lambda_{2})$ are smaller or equal to zero for every
$i,j\in L$.

\subsection{Proof of Proposition~\ref{THM:lambda.diag.solution}}\label{SEC:proof.diagonal.sol}
For a positive definite matrix $\Theta$ we denote by $\bar{\Theta}$ the diagonal matrix with the same diagonal as $\Theta$, that is $\bar{\Theta}=\diag(\Theta)$, and note that $\bar{\Theta}$ is positive definite by construction. Hence, we now show that if $\lambda_{1}\geq |s_{ij}|$ for every $i\neq j\in V$ then for  any positive definite matrix $\Theta$ it holds that
$\mathcal{L}_{\lambda_{1},\lambda_{2}}(\bar{\Theta})\leq \mathcal{L}_{\lambda_{1},\lambda_{2}}(\Theta)$ so that the matrix which minimises $\mathcal{L}_{\lambda_{1},\lambda_{2}}(\cdot)$ in (\ref{EQN:pdglasso}) must be diagonal. Concretely, we show that
$\mathcal{L}_{\lambda_{1},\lambda_{2}}(\bar{\Theta})- \mathcal{L}_{\lambda_{1},\lambda_{2}}(\Theta)\leq 0$ and, to this aim, we analyse, in turn, three components of this difference.

We first show that $\{-\log\det(\bar{\Theta})\}-\{-\log\det(\Theta)\}\leq 0$. The latter is equivalent to $\log\det(\Theta)\leq \log\det(\bar{\Theta})$, and therefore to $\det(\Theta)\leq \det(\bar{\Theta})$, which follows immediately from Hadamard's inequality; see, for instance, \citet[][p.~54 eqn (3)]{lutkepohl1996handbook}.

We now show that $\tr(S\bar{\Theta})-\tr(S\Theta)+\lambda_1\!\norm{\bar{\Theta}}_1-\lambda_1\!\norm{\Theta}_1\leq 0$. Firstly, we notice that
\begin{align*}
\tr(S\bar{\Theta})-\tr(S\Theta)
=\tr(S\bar{\Theta}-S\Theta)
=\tr\{S(\bar{\Theta}-\Theta)\}
=-\sum_{i=1}^{p}\sum_{\underset{j\neq i}{j=1}}^{p} s_{ij}\theta_{ij}
\leq \sum_{i=1}^{p}\sum_{\underset{j\neq i}{j=1}}^{p} |s_{ij}||\theta_{ij}|,
\end{align*}
where the third equality follows form the fact that both $S$ and $(\bar{\Theta}-\Theta)$ are symmetric matrices and, more specifically, that  $(\bar{\Theta}-\Theta)$ is the matrix obtained by setting to zero the diagonal entries of $-\Theta$. Similarly,
\begin{align*}
\lambda_1\!\norm{\bar{\Theta}}_1-\lambda_1\!\norm{\Theta}_1
=-\sum_{i=1}^{p}\sum_{\underset{j\neq i}{j=1}}^{p} \lambda_{1}|\theta_{ij}|;
\end{align*}
so that,
\begin{align*}
\tr(S\bar{\Theta})-\tr(S\Theta)+\lambda_1\!\norm{\bar{\Theta}}_1-\lambda_1\!\norm{\Theta}_1
&=-\sum_{i=1}^{p}\sum_{\underset{j\neq i}{j=1}}^{p} s_{ij}\theta_{ij}-\sum_{i=1}^{p}\sum_{\underset{j\neq i}{j=1}}^{p} \lambda_{1}|\theta_{ij}|\\
&\leq \sum_{i=1}^{p}\sum_{\underset{j\neq i}{j=1}}^{p} |s_{ij}||\theta_{ij}|-\sum_{i=1}^{p}\sum_{\underset{j\neq i}{j=1}}^{p} \lambda_{1}|\theta_{ij}|\\
&= \sum_{i=1}^{p}\sum_{\underset{j\neq i}{j=1}}^{p} (|s_{ij}|-\lambda_{1}) |\theta_{ij}|\leq 0,
\end{align*}
where the latter inequality holds true because, by assumption, $(|s_{ij}|-\lambda_{1})\leq 0$ for every $i\neq j\in V$.

Finally, we show that the difference
\begin{align}\label{EQN:proof1.third.part}
\lambda_2
\left(
\norm{\bar{\Theta}_{LL}-\bar{\Theta}_{RR}}_1
+\norm{\bar{\Theta}_{LR}-\bar{\Theta}_{RL}}_1
\right)
-\lambda_2
\left(
\norm{\Theta_{LL}-\Theta_{RR}}_1
+\norm{\Theta_{LR}-\Theta_{RL}}_1
\right)
\end{align}
is smaller or equal to zero. We first note that $\norm{\bar{\Theta}_{LR}-\bar{\Theta}_{RL}}_1=0$ because $\bar{\Theta}$ is diagonal and, next, that $\norm{\bar{\Theta}_{LL}-\bar{\Theta}_{RR}}_1-\norm{\Theta_{LL}-\Theta_{RR}}_1=-\sum_{i=1}^{p}\sum_{\underset{j\neq i}{j=1}}^{p} |\theta_{ij}-\theta_{i^{\prime}j^{\prime}}|$. Hence, (\ref{EQN:proof1.third.part}) can be written as
\begin{align*}
-\sum_{i=1}^{q}\sum_{\underset{j\neq i}{j=1}}^{q} \lambda_{2}|\theta_{ij}-\theta_{i^{\prime}j^{\prime}}|
-\sum_{i=1}^{q}\sum_{j=1}^{q}  \lambda_{2}|\theta_{ij^{\prime}}-\theta_{i^{\prime}j}|
\end{align*}
that is trivially non-positive because $\lambda_{2}\geq 0$, and this completes the proof.

\subsection{Proof of Proposition~\ref{THM:lambda.block.solution}}\label{SEC:proof.block.diagonal.sol}
%
For a positive definite matrix $\Theta$ we let $\bar{\Theta}$ be the block diagonal matrix with blocks $\Theta_{LL}$ and $\Theta_{RR}$. We remark that $\bar{\Theta}$ is positive definite by construction, and that the condition  $|s_{ij}|\leq \lambda_{1}$ for every $i\in L$ and $j\in R$ can be written as  $|s_{ij^{\prime}}|\leq \lambda_{1}$ for every $i,j\in L$. Hence, we now show that if $\lambda_{1}\geq |s_{ij^{\prime}}|$ for every $i,j\in L$ then for  any positive definite matrix $\Theta$ it holds that
$\mathcal{L}_{\lambda_{1},\lambda_{2}}(\bar{\Theta})\leq \mathcal{L}_{\lambda_{1},\lambda_{2}}(\Theta)$ so that the matrix which minimises $\mathcal{L}_{\lambda_{1},\lambda_{2}}(\cdot)$ in (\ref{EQN:pdglasso}) must be block diagonal with $\widehat{\Theta}_{LR}$ equal to zero. Concretely, we show that
$\mathcal{L}_{\lambda_{1},\lambda_{2}}(\bar{\Theta})- \mathcal{L}_{\lambda_{1},\lambda_{2}}(\Theta)\leq 0$ and, to this aim, we analyse, in turn, three components of this difference.

We first show that $\{-\log\det(\bar{\Theta})\}-\{-\log\det(\Theta)\}\leq 0$. The latter is equivalent to $\log\det(\Theta)\leq \log\det(\bar{\Theta})$, and therefore to $\det(\Theta)\leq \det(\bar{\Theta})$, which follows immediately from Fischers's inequality; see, for instance, \citet[][p.~54 eqn (5)]{lutkepohl1996handbook}.

We now show that $\tr(S\bar{\Theta})-\tr(S\Theta)+\lambda_1\!\norm{\bar{\Theta}}_1-\lambda_1\!\norm{\Theta}_1\leq 0$. Firstly, we notice that
\begin{align*}
\tr(S\bar{\Theta})-\tr(S\Theta)
=\tr(S\bar{\Theta}-S\Theta)
=\tr\{S(\bar{\Theta}-\Theta)\}
=-2\sum_{i=1}^{q} \sum_{j=1}^{q} s_{ij^{\prime}}\theta_{ij^{\prime}}
\leq 2\sum_{i=1}^{q} \sum_{j=1}^{q}  |s_{ij^{\prime}}||\theta_{ij^{\prime}}|,
\end{align*}
where the third equality follows form the fact that both $S$ and $(\bar{\Theta}-\Theta)$ are symmetric matrices and, more specifically, that  both $(\bar{\Theta}-\Theta)_{LL}$ and $(\bar{\Theta}-\Theta)_{RR}$ are equal to zero. Similarly,
\begin{align*}
\lambda_1\!\norm{\bar{\Theta}}_1-\lambda_1\!\norm{\Theta}_1
=- 2\sum_{i=1}^{q} \sum_{j=1}^{q}  \lambda_{1}|\theta_{ij^{\prime}}|.
\end{align*}
Hence,
\begin{align*}
\tr(S\bar{\Theta})-\tr(S\Theta)+\lambda_1\!\norm{\bar{\Theta}}_1-\lambda_1\!\norm{\Theta}_1
&=-2\sum_{i=1}^{q} \sum_{j=1}^{q} s_{ij^{\prime}}\theta_{ij^{\prime}}- 2\sum_{i=1}^{q} \sum_{j=1}^{q}  \lambda_{1}|\theta_{ij^{\prime}}|\\
&\leq 2\sum_{i=1}^{q} \sum_{j=1}^{q}  |s_{ij^{\prime}}||\theta_{ij^{\prime}}|- 2\sum_{i=1}^{q} \sum_{j=1}^{q}  \lambda_{1}|\theta_{ij^{\prime}}|\\
&= 2\sum_{i=1}^{q} \sum_{j=1}^{q} (|s_{ij^{\prime}}|-\lambda_{1}) |\theta_{ij^{\prime}}|\leq 0
\end{align*}
where the latter inequality holds true because, by assumption, $(|s_{ij^{\prime}}|-\lambda_{1})\leq 0$ for every $i,j\in L$.

In order to complete the proof we have now to show that the difference
\begin{align}\label{EQN:proof2.third.part}
\lambda_2
\left(
\norm{\bar{\Theta}_{LL}-\bar{\Theta}_{RR}}_1
+\norm{\bar{\Theta}_{LR}-\bar{\Theta}_{RL}}_1
\right)
-\lambda_2
\left(
\norm{\Theta_{LL}-\Theta_{RR}}_1
+\norm{\Theta_{LR}-\Theta_{RL}}_1
\right)
\end{align}
is smaller or equal to zero. This can be shown by noting that $\norm{\bar{\Theta}_{LR}-\bar{\Theta}_{RL}}_1=0$ because $\bar{\Theta}$ is block diagonal, and that, by construction,
$\norm{\bar{\Theta}_{LL}-\bar{\Theta}_{RR}}_{1}=\norm{\Theta_{LL}-\Theta_{RR}}_{1}$ so that
$\norm{\bar{\Theta}_{LL}-\bar{\Theta}_{RR}}_1-\norm{\Theta_{LL}-\Theta_{RR}}_1=0$. Hence (\ref{EQN:proof2.third.part}) simplifies to   $-\lambda_2\norm{\Theta_{LR}-\Theta_{RL}}_1$ that is trivially non-positive because $\lambda_{2}\geq 0$.

\section{Examples of \PDRCON\ models}\label{SEC:appendix.examples}
\newcommand{\C}{C.}
\renewcommand{\theequation}{\C\arabic{equation}}
\renewcommand{\thetable}{\C\arabic{table}}
\renewcommand{\thefigure}{\C\arabic{figure}}
\setcounter{equation}{0}

We present here some examples of \PDRCON\ models with a detailed description  of the different types of symmetry they include and of the relevant equality constraints. All the models considered involve $p=6$ variables so that  $L=\{1, 2, 3\}$ and $R=\{1^{\prime}, 2^{\prime}, 3^{\prime}\}$ and, for each of them, we give both the coloured graph and the concentration matrix. Recall that the coloured graph representing a \PDRCON\ model may contain two types of vertices and edges, namely, coloured and uncoloured.  In order to make our graphs readable also in black and white printing, the colour white is used to denote uncoloured vertices whereas coloured vertices are in gray. On the other hand, uncoloured and coloured edges are represented by thin and thick black lines, respectively. Shaded areas are used to highlight the subgraphs $\G_{L}$ and $\G_{R}$ relative to the two groups. Finally, in our representation of the concentration matrices, only the upper-triangular part is given and the entries involved in equality constraints are in bold.

%
\begin{exmp}[Figure~\ref{FIG:app.example.01}]\label{EXA:app.1}
\begin{figure}[htb]
  \centering \includegraphics[width=0.9\textwidth]{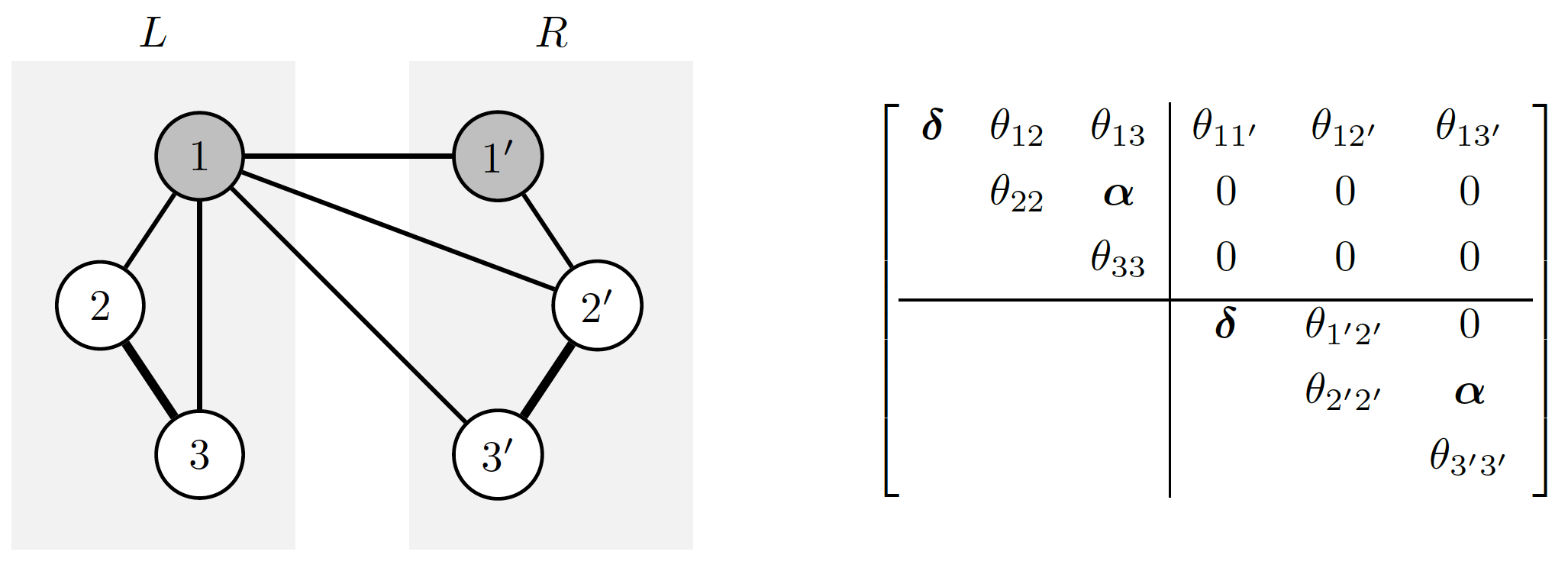}
  \caption{Coloured graph and concentration matrix of the \PDRCON\ model in Example~\ref{FIG:app.example.01}}
   \label{FIG:app.example.01}
\end{figure}
In the \PDRCON\ model of Figure~\ref{FIG:app.example.01} the edges $\{2, 3\}\in \G_{L}$ and $\{2^{\prime}, 3^{\prime}\}\in\G_{R}$ are both present and such that $\theta_{23}=\theta_{2^{\prime}3^{\prime}}=\alpha$; that is, they form an inside-block parametric symmetry. On the other hand the edges $\{1, 2\}\in \G_{L}$ and $\{1^{\prime}, 2^{\prime}\}\in \G_{R}$ are both present, thereby forming an inside-block structural symmetry, but
not a parametric symmetry because
the corresponding concentrations $\theta_{12}$ and $\theta_{1^{\prime}2^{\prime}}$ are not constrained to be equal. In this model there is also one vertex symmetry, encoded by the equality constraints $\theta_{11}=\theta_{1^{\prime}1^{\prime}}=\delta$, whereas
the only existing across-block symmetries  are those that involve missing edges, and thus zero concentrations.
\end{exmp}

\begin{exmp}[Figure~\ref{FIG:app.example.02}]\label{EXA:app.2}
\begin{figure}[htb]
  \centering  \includegraphics[width=0.95\textwidth]{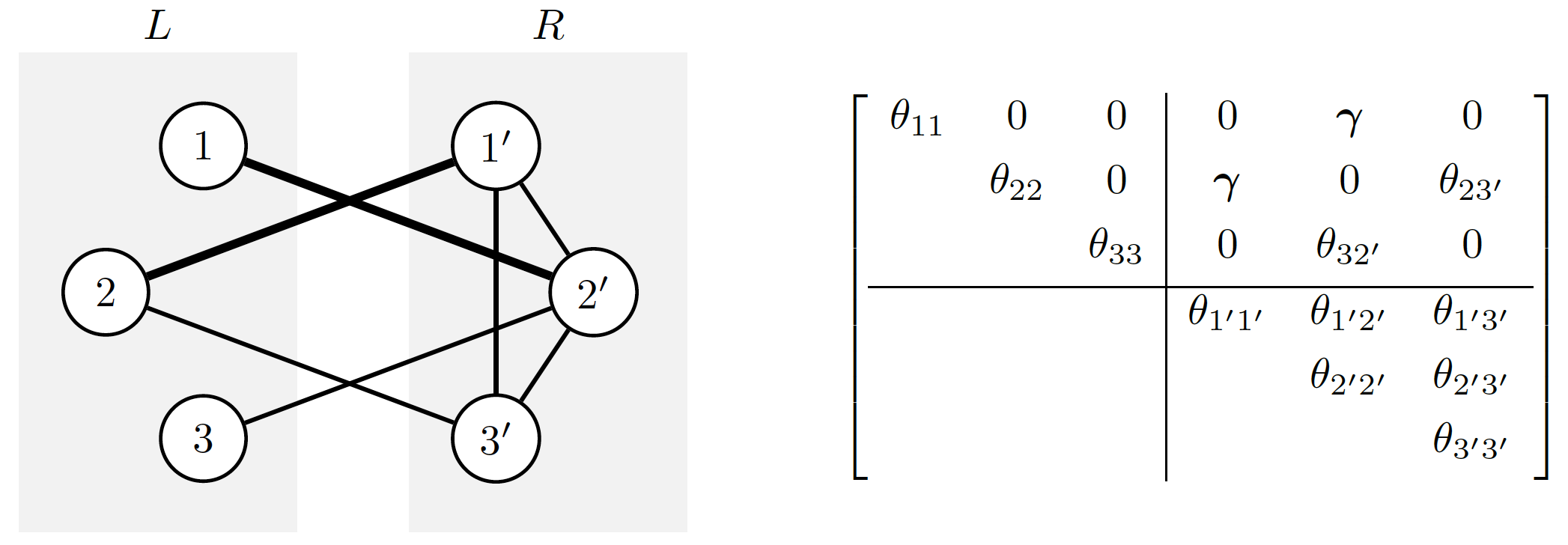}
  \caption{Coloured graph and concentration matrix of the \PDRCON\ model in Example~\ref{FIG:app.example.02}}\label{FIG:app.example.02}
\end{figure}
In the \PDRCON\ model of Figure~\ref{FIG:app.example.02} the edges $\{1, 2^{\prime}\}$ and $\{2, 1^{\prime}\}$ are both present and
form the inside-block parametric symmetry with $\theta_{12^{\prime}}=\theta_{21^{\prime}}=\gamma$. On the other hand, the edges
$\{2, 3^{\prime}\}$ and $\{3, 2^{\prime}\}$  are both present, thereby forming an across-block structural symmetry, but not a parametric symmetry because the corresponding concentrations $\theta_{23^{\prime}}$ and $\theta_{32^{\prime}}$ are not constrained to be equal. In this model, there are neither vertex symmetries nor inside-block symmetries. Indeed, the inner structure of the two groups is very different because $\G_{L}$ is a fully disconnected whereas $\G_{R}$ is complete.
\end{exmp}

\begin{exmp}[Figure~\ref{FIG:app.example.03}]\label{EXA:app.3}
\begin{figure}[htb]
  \centering \includegraphics[width=0.95\textwidth]{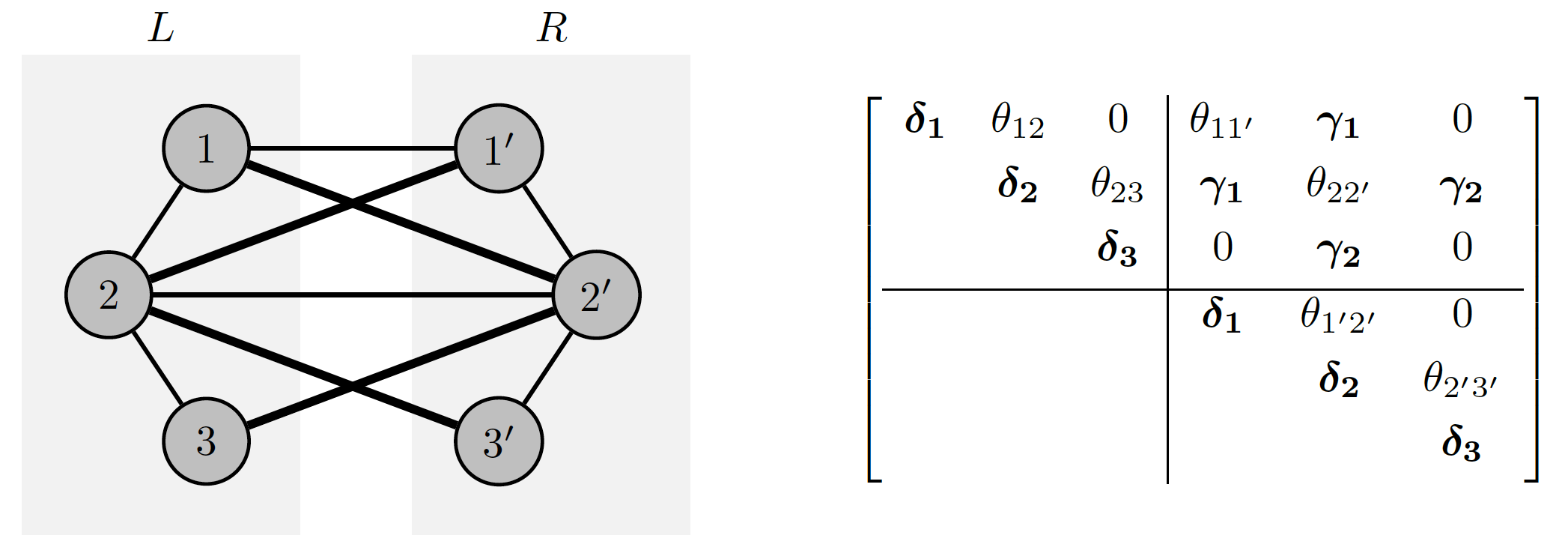}
  \caption{Coloured graph and concentration matrix of the \PDRCON\ model in Example~\ref{FIG:app.example.03}}\label{FIG:app.example.03}
\end{figure}
The \PDRCON\ model of Figure~\ref{FIG:app.example.03} shows a large amount of symmetry. More specifically, there is (i) full vertex symmetry with $\theta_{11}=\theta_{1^{\prime}1^{\prime}}=\delta_{1}$,
$\theta_{22}=\theta_{2^{\prime}2^{\prime}}=\delta_{2}$ and
$\theta_{33}=\theta_{3^{\prime}3^{\prime}}=\delta_{3}$;
(ii) full across-block parametric symmetry with $\theta_{12^{\prime}}=\theta_{2 1^{\prime}}=\gamma_{1}$ and
$\theta_{23^{\prime}}=\theta_{3^{\prime}2}=\gamma_{2}$ and, finally, (iii) there is full inside-block structural symmetry. However, there are no inside-block parametric symmetries involving present edges, and therefore this is not a fully symmetric model.
\end{exmp}

\begin{exmp}[Figure~\ref{FIG:app.example.04}]\label{EXA:app.4}
\begin{figure}[htb]
  \centering \includegraphics[width=0.95\textwidth]{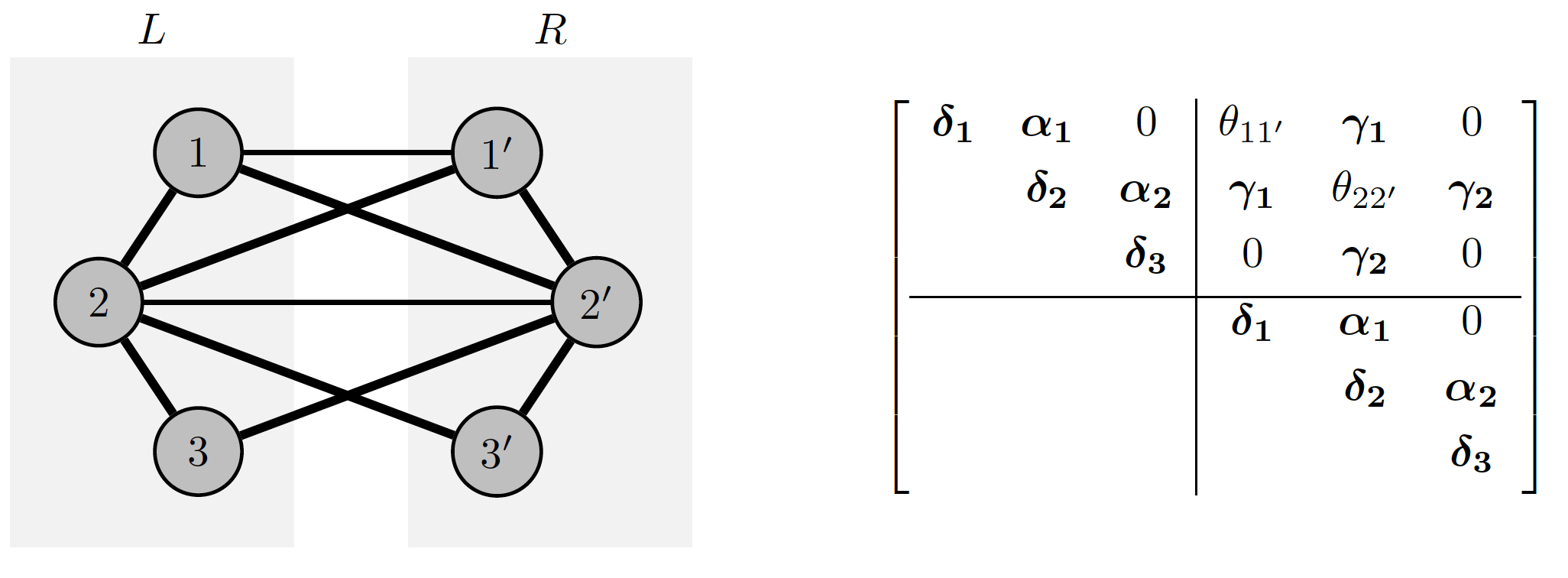}
  \caption{Coloured graph and concentration matrix of the \PDRCON\ model in Example~\ref{FIG:app.example.04}}\label{FIG:app.example.04}
\end{figure}
The structure of the graph in Figure~\ref{FIG:app.example.04} is the same as that considered in Example~\ref{EXA:app.3}, but in this case there are additional equality constraints so that full parametric symmetry is achieved. In may be worth noting that, although this is a fully symmetric model, thin lines are used to depict the edges $\{1, 1^{\prime}\}$ and $\{2, 2^{\prime}\}$. Indeed, these edges are associated with the  diagonal entries of $\Theta_{LR}$ and, in \PDRCON\ models,  these parameters are not considered for possible equality constraints.
\end{exmp}

\section{Plots of the two selected models for the breast cancer analysis of Section~\ref{SEC:application}}\label{SEC:appendix.plots}
\newcommand{\D}{D.}
\renewcommand{\theequation}{\D\arabic{equation}}
\renewcommand{\thetable}{\D\arabic{table}}
\renewcommand{\thefigure}{\D\arabic{figure}}
\setcounter{figure}{0}

Visual depiction of the two graphs associated to: (i) pdRCON model with forced full vertex symmetry and penalization on inside- and across- blocks (\texttt{mod$_{\text{fV}}$}, Figure~\ref{FIG:appendix.plots.fV}); (ii) fully symmetric pdRCON model  (\texttt{mod$_{\text{fVIA}}$}, Figure~\ref{FIG:appendix.plots.fVIA}). In both plots, the character coding is the following:  \emph{empty circle}, for structural non-zero symmetries (unicode character U+25CB, white circle); \emph{full circle}, for parametric non-zero symmetries (unicode character U+25CF, black circle); \emph{diagonal slash} for \textbf{a}symmetric edges (unicode character U+2571, box drawings light diagonal upper right to lower left). Labels for columns and rows end with the string ``\textunderscore T'' if the genes are measured on tumor tissues and ``\textunderscore H'' if measured on healthy tissues. Lower triangular portion is not shown due to the symmetric nature of the matrix.

\begin{figure}[ht!]
\centering
\includegraphics[trim={3cm 3cm 0 3cm}, clip, width=1\textwidth]{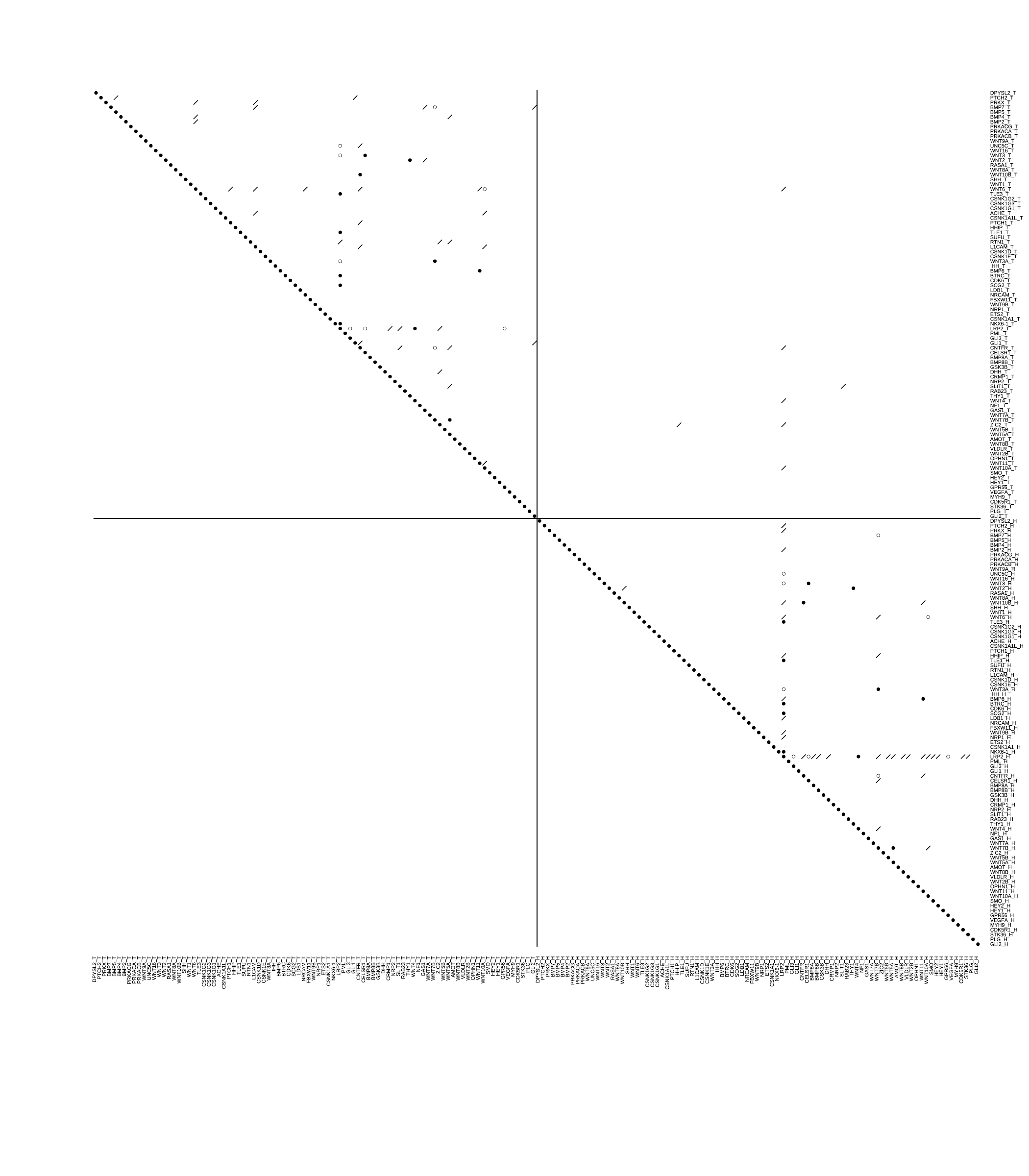}
\caption{Depiction of the graph associated to model \texttt{mod$_{\text{fV}}$}, selected according to eBIC ($\gamma=0.5$), and fit on the breast cancer data of Section~\ref{SEC:application}; for row and column labels, ``\textunderscore T'' denotes tumor tissue and ``\textunderscore H'' healthy tissue; character coding: \emph{empty circle}, structural non-zero symmetries; \emph{full circle}, parametric non-zero symmetries; \emph{diagonal slash}, \textbf{a}symmetric edges.}
\label{FIG:appendix.plots.fV}
\end{figure}

\begin{figure}[ht!]
\centering
\includegraphics[trim={3cm 3cm 0 3cm}, clip, width=1\textwidth]{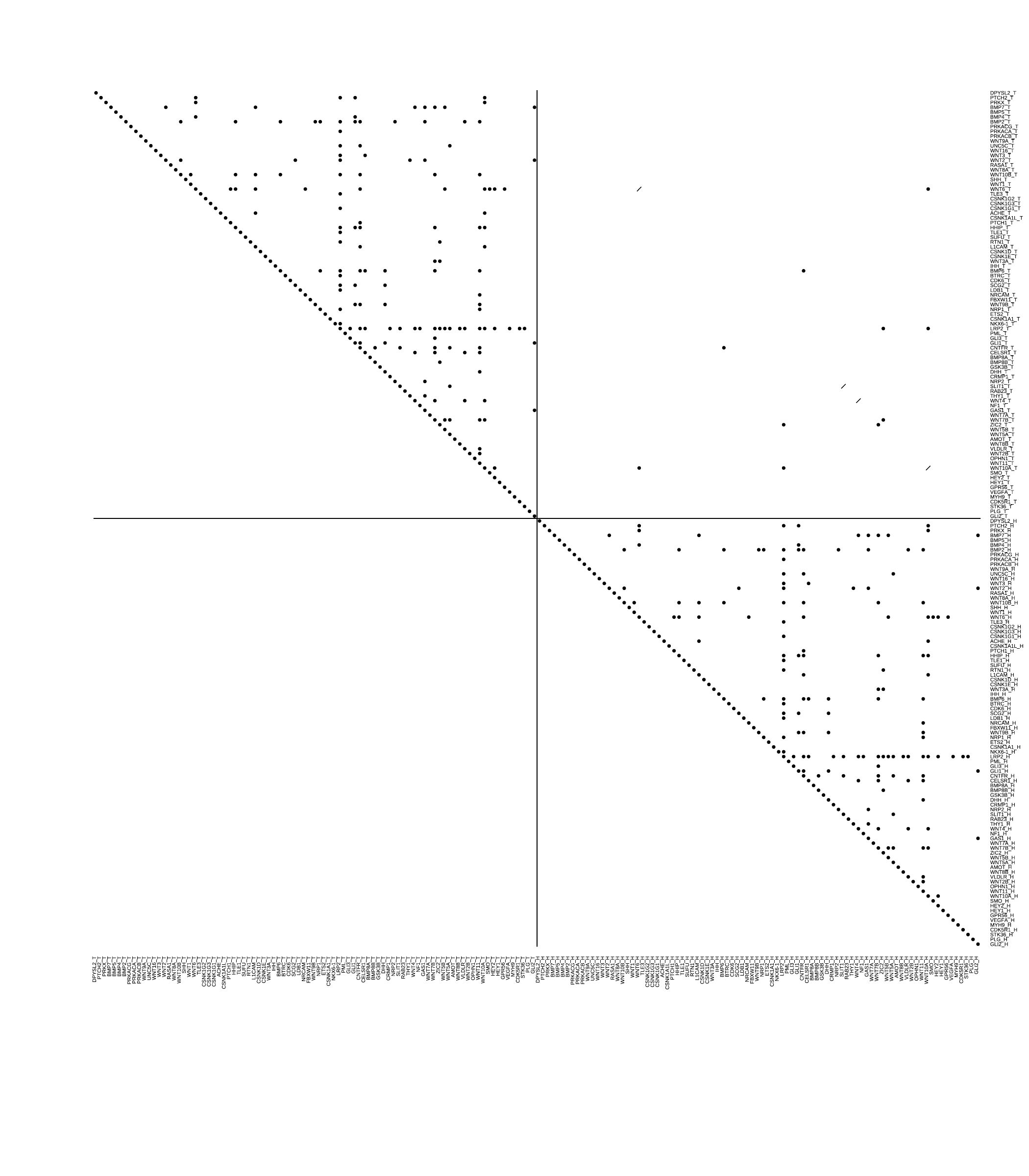}
\caption{Depiction of the graph associated to model \texttt{mod$_{\text{fVIA}}$} (fully symmetric model), selected according to eBIC ($\gamma=0.5$), and fit on the breast cancer data of Section~\ref{SEC:application}; for row and column labels, ``\textunderscore T'' denotes tumor tissue and ``\textunderscore H'' healthy tissue; character coding: \emph{empty circle}, structural non-zero symmetries; \emph{full circle}, parametric non-zero symmetries; \emph{diagonal slash}, \textbf{a}symmetric edges.}
\label{FIG:appendix.plots.fVIA}
\end{figure}

\end{document}